\documentclass[runningheads]{llncs}
\usepackage{graphicx}
\usepackage{amsmath, amssymb}
\usepackage{amsfonts}
\usepackage{url}
\usepackage{bm}
\usepackage{stmaryrd}
\usepackage{enumitem}
\usepackage{dsfont}
\usepackage{bussproofs}
\usepackage{tikz}
\usepackage{ragged2e}
\usepackage{xhfill}
\usepackage{listings}
\usepackage[linesnumbered,ruled,vlined]{algorithm2e} 

\usepackage{color}

\usepackage{hyperref}

\def\K{\mathbf{K}}
\def\LLL{\mathbf{L}}
\def\CD{\mathbf{CD}}
\def\wCD{\mathbf{wCD}}
\def\DP{\mathbf{DP}}
\def\NotPosBot{\mathbf{N}}
\def\NEC{\mathbf{NEC}}

\def\MP{\mathbf{MP}}

\def\+#1{\mathcal{#1}}
\def\-#1{\mathbf{#1}}
\newcommand{\logic}{\textbf{FIK}}

\newcommand{\calculus}{$\mathbf{C}_{\logic}$}

\newcommand{\Derivable}{$\mathbf{D}_{\logic}$}
\newcommand{\N}{\mathbb{N}}
\newcommand{\hide}[1]{}
\newcommand{\seq}{\Rightarrow}

\usepackage[T1]{fontenc}
\usepackage{graphicx}
\usepackage{color}

\title{A Natural Intuitionistic Modal Logic: Axiomatization and Bi-nested Calculus}

\author{
Philippe Balbiani\inst{1}\and
Han Gao\inst{2}\and
Çi\u{g}dem Gencer\inst{1}\and
Nicola Olivetti\inst{2}
}
\authorrunning{Balbiani, Gao, Gencer and Olivetti}
\institute{CNRS-INPT-UT3, IRIT, Toulouse, France\\
\email{\{philippe.balbiani, cigdem.gencer\}@irit.fr}
\and
Aix Marseille University, CNRS, LIS, Marseille, France\\
\email{\{gao.han, nicola.olivetti\}@lis-lab.fr}}

\titlerunning{A Natural Intuitionistic Modal Logic} 

\begin{document}
\maketitle              
\begin{abstract}
    We introduce \logic, a natural intuitionistic modal logic specified by Kripke models satisfying the condition of forward confluence. 
    We give a complete Hilbert-style axiomatization of this logic and propose a bi-nested calculus for it. 
    The calculus provides a decision procedure as well as a countermodel extraction: from any failed derivation of a given formula, we obtain by the calculus a finite  countermodel of it. 

\end{abstract}
\section{Introduction}

Intuitionistic modal logic (\textbf{IML}) has a long history,  starting from the pioneering work by Fitch~\cite{Fitch:1948} in the late 40's and Prawitz~\cite{Prawitz:1965} in the 60's. Along the time, two traditions emerged that led to the study of two different families of systems.
The first tradition, called Intuitionistic modal logics, has been introduced by Fischer Servi~\cite{Fischer:Servi:1977,Fischer:Servi:1978,Fischer:Servi:1984}, Plotkin and Stirling~\cite{Plotkin:Stirling:1986} and then systematized by Simpson~\cite{Simpson:1994}. Its main goal is to define an analogous of classical modalities justified from an intuitionistic meta-theory. 
The basic  modal logic in this tradition, \textbf{IK}, is intended to be the intuitionistic counterpart of the minimal normal modal logic \textbf{K}.
The second tradition leads to so-called Constructive modal logics that are mainly motivated by their applications in computer science such as type-theoretic interpretations,  verification and knowledge representation (contextual reasoning), together with their mathematical semantics.
This second tradition has been developed independently, first by Wijesekera~\cite{Wijesekera:1990} who proposed the system \textbf{CCDL} (Constructive Concurrent Dynamic logic), and then by Bellin, De Paiva, and Ritter~\cite{Bellin:et:al:2001}, among others who proposed the logic \textbf{CK} (Constructive \textbf{K}) as the basic system for a constructive account of modality.

But putting aside the historical perspective, we can consider naively the following question: how can we build "from scratch" an \textbf{IML}? Since both modal logic and intuitionistic logic enjoy Kripke semantics, we can think of combining them together in order to define an intuitionistic modal logic. The simplest proposal is to consider Kripke models equipped with two relations,  $\leq$ for intuitionistic implication and $R$ for modalities. Propositional intuitionistic connectives (in particular implication) have their usual interpretations. We request that every valid formula or rule scheme of propositional intuitionistic logic \textbf{IPL} is also valid in \textbf{IML}. To reach this goal, we must ensure the {\em hereditary property}, which means for any formula $A$,
\begin{quote}
	if $x\Vdash A$ and $x\leq y$ then also $y\Vdash A$.
\end{quote}
Thus the question becomes how to define modalities in order to ensure this property. The simplest solution is to build the hereditary property in the forcing conditions for $\Box$ and $\Diamond$:
\begin{quote}
(1) $x\Vdash \Box A$ iff for all $x'$ with $x'\geq x$, for all $y$ with $Rx'y$ it holds $y \Vdash A$ and\\
(1') $x\Vdash \Diamond A$ iff for all $x'$ with $x'\geq x$, there exists $y$ with $Rx'y$ s.t. $y \Vdash A$.
\end{quote}
Observe that the definition of $\Box A$ is reminiscent of the definition of $\forall$ in intuitionistic first-order logic. 
This logic is nothing else than  the propositional part of Wijeskera's  \textbf{CCDL} mentioned above and is \emph{non-normal} as it does not contain all formulas of the form
$$(DP) \ \Diamond (A\lor B) \supset \Diamond A \lor \Diamond B.$$
Moreover, the logic does not satisfy the maximality criteria, one of the criteria stated by Simpson~\cite[Chapter~$3$]{Simpson:1994} for a "good" \textbf{IML} since by adding any classical principle to it, we cannot get classical normal modal logic \textbf{K}. In addition, \textbf{CCDL} has also been criticized for being \emph{too strong}, as it still satisfies the \emph{nullary} $\Diamond$ distribution: $\Diamond \bot\supset\bot$. By removing this last axiom, the constructive modal logic \textbf{CK} is obtained. 

However,  the opposite  direction is also possible: we can make local the definition of $\Diamond$ (pursuing the analogy with $\exists$ in intuitionistic first-order logic \textbf{FOIL}) exactly as in classical \textbf{K}, that is:
\begin{quote}
(2)	$x\Vdash \Diamond A$ iff there exists $y$ with $Rxy$ s.t. $y \Vdash A$.
\end{quote}
In this way we recover $\Diamond (A\lor B) \supset \Diamond A \lor \Diamond B$, making the logic \emph{normal}. But there is a price to pay: nothing ensures that hereditary property holds for $\Diamond$-formulas. In order to solve this problem, we need to postulate some frame conditions. The most natural (and maybe the weakest) condition is simply that if $x'\geq x$ and $x$ has an $R$-accessible $y$ then also $x'$ must have an $R$-accessible $y'$ which refines $y$, which means $y'\geq y$. This condition is called \emph{Forward Confluence} in~\cite{Balbiani:et:al:2021}. It is not new as it is also called (F1) by Simpson~\cite[Chapter~$3$]{Simpson:1994} and together with another frame conditions (F2) characterizes the very well-known system \textbf{IK} by Fischer-Servi and Simpson. 
Although from a meta-theoretical point of view \textbf{IK} can be justified by its standard translation in first-order intuitionistic logic, it does not seem to be the minimal system allowing the definition of modalities as in (1) and (2) above.

This paper attempts to fill the gap by studying a weaker logic whose forcing conditions are just (1) and (2) above and we assume \emph{only} Forward Confluence.
We call this logic \textbf{FIK} for \emph{forward confluenced} \textbf{IK}. As far as we know, this logic has never been studied before. And we think it is well worth being studied: it seems to be the minimal logic defined by bi-relational models with forcing conditions (1) and (2) which preserves intuitionistic validity. 

We first give a sound and complete Hilbert axiomatization of \textbf{FIK}. We show that \textbf{FIK} finds its place in the \textbf{IML}/Constructive family: it is strictly stronger than \textbf{CCDL} (whence than \textbf{CK}) and strictly weaker than \textbf{IK}. At the same time \textbf{FIK} seems acceptable to be regarded as an \textbf{IML} since it satisfies \emph{all} criteria proposed by Simpson, including the one about maximality: by adding any classical principle to \textbf{FIK}, we get classical normal modal logic \textbf{K}. 
All in all \textbf{FIK} seems to be a respectable intuitionistic modal logic and is a kind of "third way" between intuitionistic \textbf{IK} and constructive \textbf{CCDL}/\textbf{CK}.

We then investigate \textbf{FIK} from a proof-theoretic viewpoint. We propose a nested sequent calculus \calculus \  which makes use of two kinds of nesting: one for representing $\geq$-upper worlds and the other for $R$-related worlds. A nested sequent calculus for (first-order) intuitionistic logic that makes use of the first type of nesting has been proposed in \cite{Fitting:2014}, so that our calculus can be seen as an extension of the propositional part of it. More recently in \cite{anupam:2023}, the authors present a sequent calculus with the same kind of nesting to capture the \textbf{IML} logic given by $\textbf{CCDL}+(DP)$. 

As mentioned, our calculus contains a double type of nesting. The use of this double nesting is somewhat analogous to the labelled calculus proposed in~\cite{Marin:et:al:2021} which introduces the two relations on labels in the syntax. However, the essential ingredient of the calculus \calculus \ is the {\em interaction rule} between the two kinds of nested sequents that captures the specific Forward Confluence condition.

We prove that the calculus \calculus\ provides a decision procedure for the logic \textbf{FIK}. In addition, since the rules of \calculus\ are invertible, we show that from a single failed derivation under a suitable strategy, it is possible to extract a finite countermodel of the formula or sequent at the root of the derivation. 
This result allows us to obtain a constructive proof of the finite model property, which means if a formula is not valid then it has a finite countermodel.

\section{A natural intuitionistic modal logic}\label{Sec:logic}
Firstly, we present the syntax and semantics of forward confluenced intuitionistic modal logic \logic.
Secondly, we present an axiom system and we prove its soundness and completeness.
Thirdly, we discuss whether \logic~satisfies the properties that are expected from intuitionistic modal logics.

%
%
%
%
%
%
\begin{definition}[Formulas]
The set $\+L$~of all formulas (denoted $A$, $B$, etc.) is generated by the following grammar: 
$
A::=~p~|~\bot~|~\top~|~(A\wedge A)~|~(A\vee A)~|~(A\supset A)~|~\square A~|~\Diamond A
$
where $p$ ranges over  a  countable set of atomic propositions $\textsf{At}$.
We omit parentheses for readability.
For all formulas $A$, we write $\neg A$ instead of $A\supset\bot$.
For all formulas $A,B$, we write $A\equiv B$ instead of $(A\supset B)\wedge(B\supset A)$.
The size of a formula $A$ is denoted ${\mid}A{\mid}$.
\end{definition}
\begin{definition}[Bi-relational model]\label{bi-relational-model}
 A \textit{bi-relational model} is a quadruple $\+M=(W,\leq, R,V)$ where $W$ is a nonempty set of worlds, $\leq$ is a pre-order on $W$, $R$ is a binary relation on $W$ and $V:~W\longrightarrow\wp(\textsf{At})$ is a valuation on $W$ satisfying the following hereditary condition:
 $$
\forall x,y\in W,\ (x\leq y\ \Rightarrow\ V(x)\subseteq V(y)).
$$  
The triple $(W,\leq,R)$ is called a \textit{frame.}
For all $x,y\in W$, we write $x\geq y$ instead of $y\leq x$.
Moreover, we say ``$y$ is a successor of $x$'' when $Rxy$.
\end{definition}
It is worth mentioning that
an upper world of a successor of a world is not necessarily a successor of an upper world of that world.
However, from now on in this paper, we only consider models $\+M=(W,\leq, R,V)$ that satisfy the following condition called {\em Forward Confluence} as in \cite{Balbiani:et:al:2021}:
\begin{description}
\item[(FC)] $\forall x,y\in W,\ (\exists z\in W,\ (x\geq z\ \&\ Rzy)\ \Rightarrow\ \exists t\in W,\ (Rxt\ \&\ t\geq y))$.
\end{description}
\begin{definition}[Forcing relation]\label{definition:forcing:relation}
Let $\+M=(W,\leq,R,V)$ be a bi-relational model and $w\in W$.
The forcing conditions are the usual ones for atomic propositions and for formulas constructed by means of the connectives $\bot,\top$, $\wedge$ and $\vee$.
For formulas constructed by means of the connectives $\supset$, $\square$ and $\Diamond$, the forcing conditions are as follows:
\begin{itemize}
%
%
%
%
%
%
%
%
%
%
\item $\+M,w\Vdash B\supset C$~~iff~~for all $w'\in W$ with $w\leq w'$ and $\+M,w'\Vdash B$, $\+M,w'\Vdash C$;
\item $\+M,w\Vdash \square B$~~iff~~for all $w',v'\in W$ with $w\leq w'$ and $Rw'v'$, $v'\Vdash B$;
\item $\+M,w\Vdash \Diamond B$~~iff~~there exists $v\in W$ with $Rwv$ and $\+M,v\Vdash B$.
\end{itemize}
We also abbreviate $\+M,w\Vdash A$ as $w\Vdash A$ if the model is clear from the context.
\end{definition}
\begin{proposition}\label{proposition:monotonic}
Let $(W,\leq,R,V)$ be a bi-relational model.
For all formulas $A$ in $\+L$ and for all $x,y \in W\text{~with~}x\leq y,\ x\Vdash A\text{~implies~}y\Vdash A.$
\end{proposition}
Proposition~\ref{proposition:monotonic} is proved by induction on the size of $A$ using (FC) for the case of $A=\Diamond B$.
\hide{
\begin{proof}
By induction on the complexity of $A$.
\end{proof}
}
%
%
%
%
\begin{definition}[Validity]
A formula $A$ in $\+L$ is \textit{valid}, denoted $\Vdash A$, if for any bi-relational model $\mathcal{M}$ and any world $w$ in it, $\mathcal{M},w\Vdash A$.
Let \logic~be the set of all valid formulas.
\end{definition}
%
%
%
%
%
%
Obviously, \logic~contains all standard axioms of \textbf{IPL}.
Moreover, \logic~is closed with respect to the following inference rules:
\[
  \AxiomC{$p\supset q, p$}
  \RightLabel{\rm (\textbf{MP})}
  \UnaryInfC{$q$}
  \DisplayProof
  \quad
  \AxiomC{$p$}
  \RightLabel{\rm (\textbf{NEC})}
  \UnaryInfC{$\square p$}
  \DisplayProof
\]
Finally, \logic~contains the following formulas:
\begin{description}
\item[] $(\K_{\square})$ $\square(p\supset q)\supset(\square p\supset\square q)$,
\item[] $(\K_{\Diamond})$  $\square(p\supset q)\supset(\Diamond p\supset\Diamond q)$,
\item[] $(\NotPosBot)$ $\neg\Diamond\bot$,
\item[] $(\DP)$ $\Diamond(p\vee q)\supset\Diamond p\vee\Diamond q$,
\item[] $(\wCD)$ $\square(p\vee q)\supset((\Diamond p\supset\square q)\supset\square q)$.
\end{description}
We only show the validity of $(\wCD)$.
Suppose $\not\Vdash\square(p\vee q)\supset((\Diamond p\supset\square q)\supset\square q)$.
Hence, there exists a model $(W,\leq,R,V)$ and $w\in W$ such that $w\Vdash\square(p\vee q)$, $w\Vdash\Diamond p\supset\square q$ and $w\not\Vdash\square q$.
Thus, let $u,v\in W$ be such that $w\leq u$, $Ruv$ and $v\not\Vdash q$.
Since $w\Vdash\square(p\vee q)$, $v\Vdash p\vee q$.
Since $v\not\Vdash q$, $v\Vdash p$.
Since $Ruv$, $u\Vdash\Diamond p$.
Since $w\Vdash\Diamond p\supset\square q$ and $w\leq u$, $u\Vdash\Diamond p\supset\square q$.
Since $u\Vdash\Diamond p$, $u\Vdash\square q$.
Since $Ruv$, $v\Vdash q$: a contradiction.
\begin{definition}[Axiom system]
Let \Derivable~be the Hilbert-style axiom system consisting of all standard axioms of \textbf{IPL}, the inference rules $(\MP)$ and $(\NEC)$ and the formulas $(\K_{\square})$, $(\K_{\Diamond})$, $(\NotPosBot)$, $(\DP)$ and $(\wCD)$ considered as axioms.
Derivations are defined as usual.
For all formulas $A$, we write $\vdash A$ when $A$ is \Derivable-derivable.
The set of all \Derivable-derivable formulas will also be denoted \Derivable.
\end{definition}
The formulas $(\K_{\square})$, $(\K_{\Diamond})$, $(\DP)$ and $(\NotPosBot)$ are not new, seeing that they have already been used by many authors as axioms in multifarious variants of \textbf{IML}.
As for the formula $(\wCD)$, as far as we are aware, it is used here for the first time as an axiom of an \textbf{IML} variant.
Indeed, $(\wCD)$ is derivable in \textbf{IK}.
Moreover, it is a weak form of the {\em Constant Domain}\/ axiom $(\CD):\ \square(p\vee q)\supset\Diamond p\vee\square q$ used in~\cite{Balbiani:et:al:2021}.
In other respect, $(\wCD)$ is derivable in \textbf{IK}, whereas it is not derivable in \textbf{CCDL}/\textbf{CK}.
As for the \textbf{IK} axiom $(\Diamond p\supset\square q)\supset\square(p\supset q)$, it is not in \logic\ as it will be also constructively shown by using the calculus presented in next section.
Therefore, we get
\textbf{CK}$\subset$\textbf{CCDL}$\subset$\logic$\subset$\textbf{IK}. 
We can consider also the logic $\textbf{CCDL}+(\textbf{DP})$ (= $\textbf{CK}+(\textbf{N})+(\textbf{DP})$) recently studied in \cite{anupam:2023}, according to the results in that paper, we get that $\textbf{CCDL}+(\textbf{DP})\subset \logic$.
%
%
%
%
%
%
\begin{theorem}[Soundness]\label{theorem:soundness:dfik}
\Derivable\ $\subseteq$\ \logic, i.e. for all formulas $A$, if $\vdash A$ then $\Vdash A$.
\end{theorem}
%
%
%
%
%
%
%
%
%
%
%
%
%
%
%
%
%
%
Theorem~\ref{theorem:soundness:dfik} can be proved by induction on the length of the derivation of $A$.
Later, we will prove the converse inclusion (Completeness) saying that \logic\ $\subseteq$\ \Derivable.
At the heart of our proof of completeness, there will be the concept of theory.
\begin{definition}[Theories]
A \textit{theory} is a set of formulas containing \Derivable~and closed with respect to $\MP$.
A theory $\Gamma$ is \textit{proper} if $\bot\not\in\Gamma$.
A proper theory $\Gamma$ is \textit{prime} if for all formulas $A,B$, if $A\vee B\in\Gamma$ then either $A\in\Gamma$, or $B\in\Gamma$.
For all theories $\Gamma$ and for all formulas $A$, let $\Gamma+A=\{B\in\+L:\ A\supset B\in\Gamma\}$ and $\square\Gamma=\{A\in\+L:\ \square A\in\Gamma\}$.
\end{definition}
Obviously, \Derivable~is the least theory and $\+L$ is the greatest theory.
Moreover,
%
%
%
%
%
%
%
%
for all theories $\Gamma$, $\Gamma$ is proper if and only if $\Gamma\not=\+L$ if and only if $\Diamond\bot\not\in\Gamma$.
%
%
%
%
%
%
%
%
%
%
\begin{lemma}\label{lemma:theory:gamma:plus:A}
For all theories $\Gamma$ and for all formulas $A$,
%
%
%
%
(i)~$\Gamma+A$ is the least theory containing $\Gamma$ and $A$;
%
%
(ii)~$\Gamma+A$ is proper if and only if $\neg A\not\in\Gamma$;
%
%
(iii)~$\square\Gamma$ is a theory.
%
%
%
%
\end{lemma}
%
%
%
%
%
%
%
%
%
%
%
%
%
%
%
%
%
%
%
%
%
%
Lemma~\ref{lemma:theory:gamma:plus:A} can be proved by using standard axioms of \textbf{IPL}, inference rules $(\MP)$ and $(\NEC)$ and axiom $\K_{\square}$.

\begin{lemma}[Lindenbaum's Lemma]\label{lemma:almost:completeness}
  Let $A$ be a formula. If $A\not\in$ \Derivable~then there exists a prime theory $\Gamma$ such that $A\not\in\Gamma$.
\end{lemma}

\begin{definition}[Canonical model]
Let $\bowtie$ be the binary relation between sets of formulas such that for all sets $\Delta,\Lambda$ of formulas, $\Delta\bowtie\Lambda$ iff for all formulas $B$, the following conditions hold: (i) if $\square B\in\Delta$ then $B\in\Lambda$ and (ii) if $B\in\Lambda$ then $\Diamond B\in\Delta$.

%
%
%
%
%
%
%
%
%
Let $(W_{c},\leq_{c},R_{c})$ be the frame such that
%
%
%
%
$W_{c}$ is the set of all prime theories,
%
%
$\leq_{c}$ is the inclusion relation on $W_{c}$ and
%
%
$R_{c}$ is the restriction of $\bowtie$ to $W_{c}$.
%
%
%
%
For all $\Gamma,\Delta\in W_{c}$, we write ``$\Gamma\geq_{c}\Delta$'' instead of ``$\Delta\leq_{c}\Gamma$''.
Let $V_{c}:\ W_{c}\longrightarrow\wp(\textsf{At})$ be the valuation on $W_{c}$ such that for all $\Gamma$ in $W_{c}$, $V_{c}(\Gamma)=\Gamma\cap\textsf{At}$.
%
%
\end{definition}
By Theorem~\ref{theorem:soundness:dfik}, $\bot\not\in$ \Derivable.
Hence, by Lemma~\ref{lemma:almost:completeness}, $W_{c}$ is nonempty.
\begin{lemma}\label{fc:canonical:frame:is:forward:confluent}
$(W_{c},\leq_{c},R_{c},V_{c})$ satisfies the frame condition {\bf (FC)}.
\end{lemma}
%
%
The proof of the completeness will be based on the following lemmas.
\begin{lemma}[Existence Lemma]\label{lemma:prime:proper:for:implication}
Let $\Gamma$ be a prime theory.
Let $B,C$ be formulas.
\begin{enumerate}
\item If $B\supset C\not\in\Gamma$ then there exists a prime theory $\Delta$ such that $\Gamma\subseteq\Delta$, $B\in\Delta$ and $C\not\in\Delta$,
\item if $\square B\not\in\Gamma$ then there exists prime theories $\Delta,\Lambda$ such that $\Gamma\subseteq\Delta$, $\Delta\bowtie\Lambda$ and $B\not\in\Lambda$,
\item if $\Diamond B\in\Gamma$ then there exists a prime theory $\Delta$ such that $\Gamma\bowtie\Delta$ and $B\in\Delta$.
\end{enumerate}
\end{lemma}
%
%
%
%
%
%
%
%
%
%
%
%

%
%
\begin{lemma}[Truth Lemma]\label{lemma:truth:lemma}
For all formulas $A$ and for all $\Gamma\in W_{c}$, $A\in\Gamma$ if and only if $\Gamma\models A$.
\end{lemma}
%
%
%
%
The proof of Lemma~\ref{lemma:truth:lemma} can be done by induction on the size of $A$.
The case when $A$ is an atomic proposition is by definition of $V_{c}$.
The cases when $A$ is of the form $\bot,\top$, $B\wedge C$ and $B\vee C$ are as usual.
The cases when $A$ is of the form $B\supset C$, $\square B$ and $\Diamond B$ use the Existence Lemma.

As for the proof of Theorem~\ref{theorem:completeness:dfik}, it can be done by contraposition.
Indeed, if $\not\vdash A$ then by Lemma~\ref{lemma:almost:completeness}, there exists a prime theory $\Gamma$ such that $A\not\in\Gamma$.
Thus, by Lemma~\ref{lemma:truth:lemma}, $\Gamma\not\models A$.
Consequently, $\not\Vdash A$.
%
%
%
%
%
%
\begin{theorem}[Completeness]\label{theorem:completeness:dfik}
\logic\ $\subseteq$\ \Derivable, i.e. for all formulas $A$, if $\Vdash A$ then $\vdash A$.
\end{theorem}
%
%
%
%
%
%
%
%
%
%

%
%
As mentioned above, there exists many variants of \textbf{IML}.
Therefore, one may ask how much {\em natural}\/ is the variant we consider here.
Simpson~\cite[Chapter~$3$]{Simpson:1994} discusses the formal features that might be expected of an \textbf{IML} $\LLL$:
\begin{description}
\item[$(C_{1})$] $\LLL$ is conservative over \textbf{IPL},
\item[$(C_{2})$] $\LLL$ contains all substitution instances of \textbf{IPL} and is closed under $(\MP)$,
\item[$(C_{3})$] for all formulas $A,B$, if $A\vee B$ is in $\LLL$ then either $A$ is in $\LLL$, or $B$ is in $\LLL$,
\item[$(C_{4})$] the addition of the law of excluded middle to $\LLL$ yields modal logic \textbf{K},
\item[$(C_{5})$] $\square$ and $\Diamond$ are independent in $\LLL$.
\end{description}
The fact that \Derivable~satisfies features $(C_{1})$ and $(C_{2})$ is an immediate consequence of Theorems~\ref{theorem:soundness:dfik} and~\ref{theorem:completeness:dfik}.
The fact that \Derivable~satisfies feature $(C_{3})$ will be proved in Section~\ref{section:disjunctive:property}.
Concerning feature $(C_{4})$, let \Derivable$^{+}$~be the Hilbert-style axiom system consisting of \Derivable~plus the law $p\vee\neg p$ of excluded middle.
The set of all \Derivable$^{+}$-derivable formulas will also be denoted \Derivable$^{+}$.
Obviously, \Derivable$^{+}$~contains all substitution instances of \textbf{CPL} and is closed under $(\MP)$.
Moreover, it contains all substitution instances of $(\K_{\square})$ and is closed under $(\NEC)$.
Therefore, in order to prove that \Derivable~satisfies feature $(C_{4})$, it suffices to prove
\begin{lemma}\label{lemma:about:excluded:middle:consequences}
$\Diamond p\equiv\neg\square\neg p$ is in \Derivable$^{+}$.
\end{lemma}
The fact that \Derivable~satisfies feature $(C_{5})$ is a consequence of
\begin{lemma}\label{lemma:about:feature:5}
Let $p$ be an atomic proposition.
%
%
%
%
There exists no $\square$-free $A$ such that $\square p\equiv A$ is in \Derivable\ and
%
%
there exists no $\Diamond$-free $A$ such that $\Diamond p\equiv A$ is in \Derivable.
%
%
%
%
\end{lemma}
Consequently, \Derivable~can be considered as a natural intuitionistic modal logic.
\section{A bi-nested sequent calculus}\label{section:disjunctive:property}

In this section, we present a bi-nested calculus for \logic. The calculus is two-sided and it makes use of two kinds of nested sequents, also called blocks $\langle\cdot\rangle$ and $[\cdot]$. The former is called an \textit{implication} block and the latter a \textit{modal} block. The intuition is that implication blocks correspond to upper worlds while modal blocks correspond to $R$-successors in a bi-relational model. The calculus we present is a conservative extension (with some notational change) of the nested sequent calculus for \textbf{IPL} presented in \cite{Fitting:2014}. 

\begin{definition}[Bi-nested sequent]\label{def-nested-sequent}
  A bi-nested sequent $S$ is defined as follows:
  \begin{itemize}
    \item  $\Rightarrow$ is a bi-nested sequent (the empty sequent);
    \item $\Gamma\Rightarrow B_1,\ldots, B_k,[S_1],\ldots,[S_m],\langle T_1\rangle,\ldots,\langle T_n\rangle$ is a bi-nested sequent 
    if $S_1,\ldots,S_m$, $T_1,\ldots,\\ T_n$ are bi-nested sequents where $m,n\geq 0$, and $\Gamma$ is a finite (possibly empty) multi-set of formulas and $B_1,\ldots, B_k$ are formulas.
  \end{itemize}
\end{definition}
We use $S,T$ to denote bi-nested sequents and to simplify wording we will call bi-nested sequents simply by sequents in the rest of this paper. We denote by $|S|$ {\em the size} of a sequent $S$ intended as the length of  $S$ as a string of symbols. 

As usual with nested calculi, we need the notion of context in order to specify the rules, as they can be applied to sequents occurring inside other sequents. 
A \textit{context} is of the form $G\{\}$, in which $G$ is a part of a sequent, $\{\cdot\}$ is regarded as a placeholder that needs to be filled by  another  sequent in order to complete $G$. $G\{S\}$ is the sequent obtained by replacing the occurrence of the symbol $\{\}$ in $G\{\}$ by   the sequent $S$. 

\begin{definition}[Context]\label{def-context}
  A context $G\{\}$ is inductively defined as follows:
  \begin{itemize}
    \item $\{\}$ is a context (the empty context). 
    \item if $\Gamma\Rightarrow\Delta$ is a sequent and  $G'\{\}$ is a context then $\Gamma\Rightarrow\Delta, \langle G'\{\}\rangle$ is a context.
    \item if $\Gamma\Rightarrow\Delta$ is a sequent and  $G'\{\}$ is a context then $\Gamma\Rightarrow\Delta, [G'\{\}]$ is a context.
    \end{itemize}
\end{definition}

For example, given a context $G\{\}=A\wedge B, \square C\Rightarrow \langle \square A\Rightarrow [B]\rangle,[\{\}]$ and a sequent $S=A\Rightarrow \Delta,[C\Rightarrow B]$, we have $G\{S\}$ $=$ $A\wedge B, \square C\Rightarrow \langle \square A\Rightarrow [B]\rangle,[A\Rightarrow \Delta,[C\Rightarrow B]]$. 

The two types of blocks interact by the (inter) rule. In order to define this rule, we need the following:

\begin{definition}[$*$-operator]\label{def:star-operator}
  Let $\Lambda\Rightarrow\Theta$ be a sequent, we define $\Theta^*$ as follows:
  \begin{itemize}
    \item $\Theta^*=\emptyset$ if $\Theta$ is $[\cdot]$-free;
    \item $\Theta^*= [\Phi_1\Rightarrow\Psi_1^*], \ldots, [\Phi_k\Rightarrow\Psi_k^*]$ if $\Theta=\Theta_0,[\Phi_1\Rightarrow\Psi_1], \ldots, [\Phi_k\Rightarrow\Psi_k]$ and $\Theta_0$ is $[\cdot]$-free.
  \end{itemize}
\end{definition}
By definition, given a sequent $\Lambda\Rightarrow\Theta$, $\Theta^*$ is a multi-set of modal blocks.

Now we can give a bi-nested sequent calculus for \logic~as follows.

\begin{definition}
The calculus \calculus~is given in Figure \ref{bi-calculus}.
\begin{figure}[!htb]
    \centering
    \noindent\fbox{
      \parbox{\linewidth}{
        {
  Axioms:
  \vspace{-0.2cm}
  {
      \[ 
    \AxiomC{}
    \RightLabel{\rm ($\bot_L$)}
    \UnaryInfC{$G\{\Gamma,\bot\Rightarrow\Delta\}$}
    \DisplayProof
    \quad
    \AxiomC{}
    \RightLabel{\rm ($\top_R$)}
    \UnaryInfC{$G\{\Gamma\Rightarrow\top,\Delta\}$}
    \DisplayProof
    \quad
    \AxiomC{}
    \RightLabel{\rm ($\text{id}$)}
    \UnaryInfC{$G\{\Gamma,A\Rightarrow\Delta,A\}$}
    \DisplayProof
    \]
  }
   
  Logical rules:
   \vspace{-0.2cm}
  {
    \[
    \AxiomC{$G\{A,B,\Gamma\Rightarrow\Delta\}$}
    \RightLabel{\rm ($\wedge_L$)}
    \UnaryInfC{$G\{A\wedge B,\Gamma\Rightarrow\Delta\}$}
    \DisplayProof
    \quad
    \AxiomC{$G\{\Gamma\Rightarrow\Delta,A\}$}
    \AxiomC{$G\{\Gamma\Rightarrow\Delta,B\}$}
    \RightLabel{\rm ($\wedge_R$)}
    \BinaryInfC{$G\{\Gamma\Rightarrow\Delta,A\wedge B\}$}
  \DisplayProof
  \]
  
  \[
    \AxiomC{$G\{\Gamma,A\Rightarrow\Delta\}$}
    \AxiomC{$G\{\Gamma,B\Rightarrow\Delta\}$}
    \RightLabel{\rm ($\vee_L$)}
    \BinaryInfC{$G\{\Gamma,A\vee B\Rightarrow\Delta\}$}
    \DisplayProof
    \quad
    \AxiomC{$G\{\Gamma\Rightarrow \Delta,A,B\}$}
    \RightLabel{\rm ($\vee_R$)}
    \UnaryInfC{$G\{\Gamma\Rightarrow \Delta,A\vee B\}$}
  \DisplayProof
  \]
  
  \[
  \AxiomC{$G\{\Gamma,A\supset B\Rightarrow A,\Delta\}$}
  \AxiomC{$G\{\Gamma,B\Rightarrow\Delta\}$}
  \RightLabel{\rm ($\supset_L$)}
  \BinaryInfC{$G\{\Gamma,A\supset B\Rightarrow\Delta\}$}
  \DisplayProof
  \quad
  \AxiomC{$G\{\Gamma\Rightarrow \Delta, \langle A\Rightarrow B\rangle\}$}
  \RightLabel{\rm ($\supset_R$)}
  \UnaryInfC{$G\{\Gamma\Rightarrow \Delta, A\supset B\}$}
  \DisplayProof
  \]
  }
   
  Modal rules:
   \vspace{-0.2cm}
  \[
    \AxiomC{$G\{\Gamma,\square A\Rightarrow\Delta,[\Sigma,A\Rightarrow \Pi]\}$}
  \RightLabel{\rm ($\square_L$)}
  \UnaryInfC{$G\{\Gamma,\square A\Rightarrow\Delta,[\Sigma\Rightarrow \Pi]\}$}
  \DisplayProof
  \quad
  \AxiomC{$G\{\Gamma\Rightarrow\Delta,\langle\Rightarrow [\Rightarrow A]\rangle\}$}
  \RightLabel{\rm ($\square_R$)}
  \UnaryInfC{$G\{\Gamma\Rightarrow\Delta,\square A\}$}
  \DisplayProof
  \]
  \[
    \AxiomC{$G\{\Gamma\Rightarrow\Delta,[A\Rightarrow]\}$}
    \RightLabel{\rm ($\Diamond_L$)}
    \UnaryInfC{$G\{\Gamma,\Diamond A\Rightarrow\Delta\}$}
    \DisplayProof
    \quad
    \AxiomC{$G\{\Gamma\Rightarrow\Delta,\Diamond A,[\Sigma\Rightarrow\Pi,A]\}$}
    \RightLabel{\rm ($\Diamond_R$)}
    \UnaryInfC{$G\{\Gamma\Rightarrow\Delta,\Diamond A,[\Sigma\Rightarrow\Pi]\}$}
    \DisplayProof
  \]
  
  Transferring and interactive rules:
   \vspace{-0.2cm}
  \[
    \AxiomC{$G\{\Gamma,\Gamma'\Rightarrow\Delta,\langle\Gamma',\Sigma\Rightarrow\Pi\rangle\}$}
    \RightLabel{\rm (\text{trans})}
    \UnaryInfC{$G\{\Gamma,\Gamma'\Rightarrow\Delta,\langle\Sigma\Rightarrow\Pi\rangle\}$}
    \DisplayProof
    \quad
    \AxiomC{$G\{\Gamma\Rightarrow\Delta,\langle\Sigma\Rightarrow\Pi,[\Lambda\Rightarrow\Theta^*]\rangle,[\Lambda\Rightarrow\Theta]\}$}
    \RightLabel{\rm (\textrm{inter})}
    \UnaryInfC{$G\{\Gamma\Rightarrow\Delta,\langle\Sigma\Rightarrow\Pi\rangle,[\Lambda\Rightarrow\Theta]\}$}
    \DisplayProof
  \]
  
  }
        }
  }
  \caption{\calculus}\label{bi-calculus}
  \end{figure}
\end{definition}

Here is a brief explanation of these rules. The logical rules, except $(\supset_R)$, are just the standard rules of intuitionistic logic in their nested version. 
The rule $(\supset_R)$ introduces an implication block, which corresponds to an upper world (in the pre-order). The modal rules create new modal blocks or propagate modal formulas into existing ones, which correspond to $R$-accessible worlds. The (trans) rule transfers formulas (forced by) lower worlds to upper worlds following the pre-order. Finally, (inter) rule encodes the (FC) frame condition: it partially transfers "accessible" modal blocks from lower worlds to upper ones and creates new accessible worlds from upper worlds fulfilling the (FC) condition.

We define the modal degree of a sequent, which will be useful when discussing termination.

\begin{definition}[Modal degree]
  Modal degree for a formula $F$, denoted as $\textit{md}(F)$, is defined as usual: $md(p) = md(\bot) = md(\top) = 0$, $md(A \circ B) = max(md(A), md(B))$, for $\circ = \land, \lor, \supset$, $md(\Box A) = md(\Diamond A) = md(A)+1$. Further, if  $\Gamma =\{A_1,\ldots A_n\}$ then  $md(\Gamma) =max (md(A_1), \ldots, md(A_n))$. For a  sequent $S=\Gamma\Rightarrow\Delta,[S_1],\ldots,[S_m],\langle T_1\rangle,\ldots,\langle T_n\rangle$ with  $m,n\geq 0$, let  $md(S)=\max(md(\Gamma),md(\Delta),md(S_1)+1,\ldots,md(S_m)+1,md(T_1),\ldots,md(T_n))$.
\end{definition}

\begin{example}
  Axiom (\textbf{wCD}) in $\mathbf{S}_\logic$ is provable in \calculus. To prove this, it suffices to prove $\Diamond p\supset\square q,\square(p\vee q)\Rightarrow\square q$.
 \[
  \AxiomC{$\Diamond p\supset\square q,\square(p\vee q)\Rightarrow\langle\Diamond p\supset\square q,\square(p\vee q)\Rightarrow[\Rightarrow q]\rangle$}
  \RightLabel{$(\text{trans})$}
  \UnaryInfC{$\Diamond p\supset\square q,\square(p\vee q)\Rightarrow\langle\Rightarrow[\Rightarrow q]\rangle$}
  \RightLabel{$(\square_R)$}
  \UnaryInfC{$\Diamond p\supset\square q,\square(p\vee q)\Rightarrow\square q$}
  \DisplayProof
  \]
  
   Let $G\{\}=\Diamond p\supset\square q,\square(p\vee q)\Rightarrow\langle\{\}\rangle$, so $G\{\Diamond p\supset\square q,\square(p\vee q)\Rightarrow[\Rightarrow q]\}$ is $\Diamond p\supset\square q,\square(p\vee q)\Rightarrow\langle\Diamond p\supset\square q,\square(p\vee q)\Rightarrow[\Rightarrow q]\rangle$. Then the derivation of the topmost sequent is as follows:
  \[
  \tiny
  \AxiomC{}
  \RightLabel{$(\text{id})$}
  \UnaryInfC{$G\{\Diamond p\supset\square q,\square(p\vee q)\Rightarrow[p\Rightarrow q,p]\}$}
  \RightLabel{$(\Diamond_R)$}
  \UnaryInfC{$G\{\Diamond p\supset\square q,\square(p\vee q)\Rightarrow\Diamond p,[p\Rightarrow q]\}$}
  \AxiomC{}
  \RightLabel{$(\text{id})$}
  \UnaryInfC{$G\{\square q,\square(p\vee q)\Rightarrow[q,p\Rightarrow q]\}$}
  \RightLabel{$(\square_L)$}
  \UnaryInfC{$G\{\square q,\square(p\vee q)\Rightarrow[p\Rightarrow q]\}$}
  \RightLabel{$(\supset_L)$}
  \BinaryInfC{$G\{\Diamond p\supset\square q,\square(p\vee q)\Rightarrow[p\Rightarrow q]\}$}
  \AxiomC{}
  \RightLabel{$(\text{id})$}
  \UnaryInfC{$G\{\Diamond p\supset\square q,\square(p\vee q)\Rightarrow[q\Rightarrow q]\}$}
  \RightLabel{$(\vee_L)$}
  \BinaryInfC{$G\{\Diamond p\supset\square q,\square(p\vee q)\Rightarrow[p\vee q\Rightarrow q]\}$}
  \RightLabel{$(\square_L)$}
  \UnaryInfC{$G\{\Diamond p\supset\square q,\square(p\vee q)\Rightarrow[\Rightarrow q]\}$}
  \DisplayProof
  \] 
\end{example}

\begin{example}\label{examplefik}
Consider the formula $\Rightarrow(\neg\square \bot\supset\square\bot)\supset\square\bot$. This $\Diamond$-free formula is provable in \calculus \ but unprovable in \textbf{CK} (whence the $\Diamond$-free fragments of these two logics are different, see \cite{anupam:2023}). 
\[
    \tiny
    \AxiomC{$G\{S_1\}$}
    \UnaryInfC{$\Rightarrow\langle\neg\square\bot\supset\square\bot\Rightarrow\langle\neg\square\bot\supset\square\bot\Rightarrow\neg\square\bot,[\Rightarrow\bot]\rangle\rangle$}
    \AxiomC{$G\{S_2\}$}
    \UnaryInfC{$\Rightarrow\langle\neg\square\bot\supset\square\bot\Rightarrow\langle\square\bot\Rightarrow[\Rightarrow\bot]\rangle\rangle$}
    \RightLabel{($\supset_L$)}
    \BinaryInfC{$\Rightarrow\langle\neg\square\bot\supset\square\bot\Rightarrow\langle\neg\square\bot\supset\square\bot\Rightarrow[\Rightarrow\bot]\rangle\rangle$}
    \RightLabel{(trans)}
    \UnaryInfC{$\Rightarrow\langle\neg\square\bot\supset\square\bot\Rightarrow\langle\Rightarrow[\Rightarrow\bot]\rangle\rangle$}
    \RightLabel{($\square_R$)}
    \UnaryInfC{$\Rightarrow\langle\neg\square\bot\supset\square\bot\Rightarrow\square\bot\rangle$}
    \RightLabel{($\supset_R$)}
    \UnaryInfC{$\Rightarrow(\neg\square\bot\supset\square\bot)\supset\square\bot$}
    \DisplayProof
  \]

Let $G\{\}=~\Rightarrow\langle\neg\square\bot\supset\square\bot\Rightarrow\langle\{\}\rangle\rangle$, $S_1=\neg\square\bot\supset\square\bot\Rightarrow\neg\square\bot,[\Rightarrow\bot]$ and $S_2=\square\bot\Rightarrow[\Rightarrow\bot]$. The two top sequents $G\{S_1\}$ and $G\{S_2\}$ are derived respectively as follows:
    \[
  \AxiomC{}
  \RightLabel{$(\bot_L)$}
  \UnaryInfC{$G\{\neg\square\bot\supset\square\bot\Rightarrow\langle\square\bot\Rightarrow\bot,[\bot\Rightarrow]\rangle,[\Rightarrow\bot]\}$}
  \RightLabel{$(\square_L)$}
  \UnaryInfC{$G\{\neg\square\bot\supset\square\bot\Rightarrow\langle\square\bot\Rightarrow\bot,[\Rightarrow]\rangle,[\Rightarrow\bot]\}$}
  \RightLabel{($\text{inter}$)}
  \UnaryInfC{$G\{\neg\square\bot\supset\square\bot\Rightarrow\langle\square\bot\Rightarrow\bot\rangle,[\Rightarrow\bot]\}$}
  \RightLabel{$(\supset_R)$}
  \UnaryInfC{$G\{\neg\square\bot\supset\square\bot\Rightarrow\neg\square\bot,[\Rightarrow\bot]\}$}
  \DisplayProof
  \quad
  \AxiomC{}
  \RightLabel{$(\bot_L)$}
  \UnaryInfC{$G\{\square\bot\Rightarrow[\bot\Rightarrow\bot]\}$}
  \RightLabel{$\square_L$}
  \UnaryInfC{$G\{\square\bot\Rightarrow[\Rightarrow\bot]\}$}
  \DisplayProof
  \]
\end{example}

We show that  the calculus \calculus~enjoys the disjunctive property, which means if $A\lor B$ is provable, then either $A$ or $B$ is provable. This fact is an immediate consequence of the following lemma.

\begin{lemma}\label{lemma-sec3-dp}
	Suppose that a sequent $S =~\Rightarrow A_1, \ldots, A_m, \langle G_1\rangle, \ldots, \langle G_n\rangle$ is provable in \calculus, where $A_i$'s are formulas. Then either for some $A_i$, $\Rightarrow A_i$ is provable in \calculus~or for some $G_j$, $\Rightarrow \langle G_j\rangle$ is provable in \calculus.  
\end{lemma}

From the lemma we immediately obtain:

\begin{proposition}\label{prop:dp}
  For any formulas $A,B$, if $\Rightarrow A \lor B$ is provable in \calculus, then either $\Rightarrow A$ or $\Rightarrow  B$ is provable.
\end{proposition}
\hide{
}
By the soundness and completeness of \calculus~with respect to \logic~proved in the following, we will conclude that the logic \logic~enjoys the disjunctive property.  


Next, we prove the soundness of the calculus \calculus. To achieve this aim, we need to define the semantic interpretation of sequents, whence their validity. We first extend the forcing relation $\Vdash$ to sequents and blocks therein. 

\begin{definition}
Let $\mathcal{M}=(W,\leq,R,V)$ be a bi-relational model and $x\in W$. The relation $\Vdash$ is extended to sequents as follows:
\begin{quote}
$\mathcal{M}, x\not\Vdash \emptyset$\\
$\mathcal{M}, x \Vdash [T]$ if for every $y$ with $Rxy$, $\mathcal{M}, y \Vdash T$\\
$\mathcal{M}, x \Vdash \langle T \rangle$ if for every $x'$ with $x\leq x'$, $\mathcal{M}, x' \Vdash T$\\
$\mathcal{M}, x \Vdash \Gamma \Rightarrow \Delta$ if either $\mathcal{M}, x \not \Vdash A$ for some $A\in \Gamma$ or $\mathcal{M}, x \Vdash {\cal O}$ for some ${\cal O} \in \Delta$
\end{quote}
We say $S$ is {\rm valid} in $\+M$ iff $\forall w\in W$, we have $\+M,w\Vdash S$. $S$ is {\rm valid} iff it is valid in every bi-relational model.
\end{definition}

Whenever the model $\mathcal{M}$ is clear, we omit it and write simply $x \Vdash {\cal O}$ for any object ${\cal O}$, which can be a formula, a sequent or a block. 
Moreover, given a sequent $S = \Gamma \Rightarrow \Delta$, we write $x \Vdash \Delta$ if there is ${\cal O} \in \Delta$ s.t. $x \Vdash {\cal O}$ and write $x \not\Vdash \Delta$ if the previous condition does not hold. 

The following lemma gives a semantic meaning to the $*$-operation used in (inter). 

\begin{lemma}\label{lemma:star-operator}
  Let $\mathcal{M}=(W,\leq,R,V)$ be a bi-relational model and $x, x'\in W$ with $x\leq x'$. Let $S = \Gamma \Rightarrow \Delta$ be any sequent, if $x \not\Vdash \Delta$ then $x' \not\Vdash \Delta^*$.
\end{lemma}

In order to prove soundness we first show that the all rules are {\em forcing-preserving}.

\begin{lemma}\label{forcing-preserving}
  Given a model $\mathcal{M} = (W,  \leq, R, V)$ and $x\in W$, for any rule ($r$) of the form $\frac{G\{S_1\} \quad G\{S_2\}}{G\{S\}}$ or $\frac{G\{S_1\}}{G\{S\}}$, if $x\Vdash G\{S_i\}$, then $x\Vdash G\{S\}$. 
\end{lemma}

Proof of this lemma proceeds by induction on the structure of the context $G\{  \ \}$. The the base of the induction (that is $G=\emptyset$) is the important one, we check rule by rule and in the case of (inter) we make use of Lemma \ref{lemma:star-operator}. 

By Lemma \ref{forcing-preserving}, the soundness of \calculus~is proved as usual by a straightforward induction on the length of derivations.

\begin{theorem}[Soundness]
  If a sequent $S$ is provable in \calculus, then it is valid.
\end{theorem}
\section{Termination and completeness for \calculus}

In this section, we provide a terminating proof-search procedure based on \calculus, whence a decision procedure for \logic; it will then be used to prove that \calculus~is complete with respect to \logic~bi-relational semantics. Here is a roadmap: first we introduce a set-based variant of the calculus where all rules are cumulative (or kleen'ed), in the sense that principal formulas are kept in the premises. With this variant, we formulate saturation conditions on a sequent associated to each rule. Saturation conditions are needed for both termination and completeness: they are used to prevent "redundant" application of the rules, source of non-termination. In the meantime saturation conditions also ensure that a saturated sequent satisfies the truth conditions specified by the semantics (which is presented in truth lemma), so it can be seen as a countermodel. 

First, we present \textbf{C}\calculus, a variant of \calculus~where sequents are set-based rather than multi-set based and the rules are cumulative. 

\begin{definition}
  {\rm \textbf{C}\calculus}~acts on set-based sequents, where a set-based sequent $S = \Gamma \Rightarrow \Delta$ is defined  as in definition \ref{def-nested-sequent}, but $\Gamma$ is a {\em set} of formulas and $\Delta$ is a {\em set} of formulas and/or blocks (containing set-based sequents).
  The rules are as follows: 
  \begin{itemize}
    \item It contains the rules $(\bot_L),\ (\text{id}),\ (\square_L),\ (\Diamond_R)$, (trans) and (inter) of \calculus. 
    \item  $(\supset_R)$ is replaced by the two  rules:
    {
      \tiny
    \[
      \text{if~$A\in\Gamma$~~~~}
    \AxiomC{$G\{\Gamma\Rightarrow\Delta,A\supset B,B\}$}
    \RightLabel{\rm ($\supset_{R_1}'$)}
    \UnaryInfC{$G\{\Gamma\Rightarrow\Delta,A\supset B\}$}
    \DisplayProof
    \quad
    \text{~~~~~if~$A\notin\Gamma$~~~~}
    \AxiomC{$G\{\Gamma\Rightarrow \Delta,A\supset B,\langle A\Rightarrow B\rangle\}$}
    \RightLabel{\rm ($\supset_{R_2}'$)}
    \UnaryInfC{$G\{\Gamma\Rightarrow\Delta,A\supset B\}$}
    \DisplayProof
    \]
    }
    \item The other rules are modified in order to keep  the principal formula in the premises. For example, the cumulative versions of $(\wedge_L),\ (\supset_L),\ (\square_R)$ and $(\Diamond_L)$  are:
    {
      \tiny
    \[
      \AxiomC{$G\{A,B,A\wedge B,\Gamma\Rightarrow\Delta\}$}
    \RightLabel{\rm ($\wedge_L'$)}
    \UnaryInfC{$G\{A\wedge B,\Gamma\Rightarrow\Delta\}$}
    \DisplayProof
    \quad
    \AxiomC{$G\{\Gamma,A\supset B\Rightarrow A,\Delta\}$}
    \AxiomC{$G\{\Gamma,A\supset B,B\Rightarrow\Delta\}$}
    \RightLabel{\rm ($\supset_L'$)}
    \BinaryInfC{$G\{\Gamma,A\supset B\Rightarrow\Delta\}$}
    \DisplayProof
    \]
  
    \[
    \AxiomC{$G\{\Gamma\Rightarrow\Delta,\square A,\langle\Rightarrow [\Rightarrow A]\rangle\}$}
    \RightLabel{\rm ($\square_R'$)}
    \UnaryInfC{$G\{\Gamma\Rightarrow\Delta,\square A\}$}
    \DisplayProof
    \quad
    \AxiomC{$G\{\Gamma,\Diamond A\Rightarrow\Delta,[A\Rightarrow]\}$}
    \RightLabel{\rm ($\Diamond_L'$)}
    \UnaryInfC{$G\{\Gamma,\Diamond A\Rightarrow\Delta\}$}
    \DisplayProof
    \]
    }
  \end{itemize}
\end{definition} 

\begin{proposition}
  A sequent $S$ is provable in \calculus~if and only if $S$ is provable in {\rm \textbf{C}\calculus}.
\end{proposition}

From now on we consider \textbf{C}\calculus. 
We introduce the notion of {\em structural inclusion} between sequents.
It is used in the definition of saturation conditions as well as the model construction presented at the end of the section. 

\begin{definition}[Structural inclusion $\subseteq^\-S$]\label{def:structural-inclusion}
  Let $\Gamma_1\Rightarrow\Delta_1,\Gamma_2\Rightarrow\Delta_2$ be two sequents. $\Gamma_1\Rightarrow\Delta_1$ is said to be \textit{structurally included} in $\Gamma_2\Rightarrow\Delta_2$, denoted as $\Gamma_1\Rightarrow\Delta_1\subseteq^\-S \Gamma_2\Rightarrow\Delta_2$, if:
  \begin{itemize}
    \item  $\Gamma_1\subseteq\Gamma_2$ and ;
    \item for each $[\Lambda_1\Rightarrow\Theta_1]\in \Delta_1$,  there exists $[\Lambda_2\Rightarrow\Theta_2]\in\Delta_2$ such that
  $\Lambda_1\Rightarrow\Theta_1\subseteq^\-S \Lambda_2\Rightarrow\Theta_2$.
  \end{itemize}
\end{definition}

It is easy to see that $\subseteq^\-S$ is reflexive and transitive; moreover if $\Gamma_1\Rightarrow\Delta_1\subseteq^\-S \Gamma_2\Rightarrow\Delta_2$, then $\Gamma_1\subseteq\Gamma_2$.

We define now  the saturation conditions associated to each rule of \textbf{C}\calculus.

\begin{definition}[Saturation conditions]\label{saturation-condition}
  Let $\Gamma\Rightarrow\Delta$ be a sequent where $\Gamma$ is a set of formulas and $\Delta$ is a set of formulas and blocks. Saturation conditions associated to a rule in the calculus are given as below.
  \scriptsize
\begin{description}
  \item[($\bot_L$)] $\bot\notin\Gamma$.
  \item[($\top_R$)] $\top\notin\Delta$. 
  \item[(id)] $\mathsf{At}\cap(\Gamma\cap\Delta)$ is empty.
  \item[($\wedge_R$)] If $A\wedge B\in \Delta$, then $A\in\Delta$ or $B\in \Delta$.
  \item[($\wedge_L$)] If $A\wedge B\in \Gamma$, then $A\in \Gamma$ and $B\in\Gamma$.
  \item[($\vee_R$)] If $A\vee B\in\Delta$, then $A\in\Delta$ and $B\in\Delta$.
  \item[($\vee_L$)] If $A\vee B\in \Gamma$, then $A\in\Gamma$ or $B\in\Gamma$.
  \item[($\supset_R$)] If $A\supset B\in\Delta$, then either $A\in\Gamma$ and $B\in\Delta$, or there is $\langle\Sigma\Rightarrow\Pi\rangle\in\Delta$ with $A\in\Sigma$ and $B\in\Pi$.
  \item[($\supset_L$)] If $A\supset B\in\Gamma$, then $A\in\Delta$ or $B\in\Gamma$.
  \item[($\square_R$)] If $\square A\in \Delta$, then either there is $[\Lambda\Rightarrow\Theta]\in\Delta$ with $A\in\Theta$, or there is $\langle\Sigma\Rightarrow[\Lambda\Rightarrow\Theta],\Pi\rangle\in\Delta$ with $A\in\Theta$.
  \item[($\square_L$)] If $\square A\in\Gamma$ and $[\Sigma\Rightarrow\Pi]\in\Delta$, then $A\in \Sigma$.
  \item[($\Diamond_R$)] If $\Diamond A\in \Delta$ and $[\Sigma\Rightarrow\Pi]\in\Delta$, then $A\in\Pi$.
  \item[($\Diamond_L$)] If $\Diamond A\in\Gamma$, then there is $[\Sigma\Rightarrow\Pi]\in\Delta$ with $A\in\Sigma$.
  \item[(trans)] If $\Delta$ is of form $\Delta',\langle\Sigma\Rightarrow\Pi\rangle$, then $\Gamma\subseteq\Sigma$.
  \item[(inter)] If $\Delta$ is of form $\Delta',\langle\Sigma\Rightarrow\Pi\rangle,[\Lambda\Rightarrow\Theta]$, then there is $[\Phi\Rightarrow\Psi]\in\Pi$ with $\Lambda\Rightarrow\Theta\subseteq^\-S\Phi\Rightarrow\Psi$.
\end{description}
\end{definition}
Concerning (inter)-saturation, observe that $\Lambda\Rightarrow\Theta\subseteq^\-S \Lambda\Rightarrow\Theta^*$, thus this condition generalizes the expansion produced by the (inter)-rule. 

\begin{proposition}\label{sequent-inclusion}
  Let $\Gamma\Rightarrow\Delta$ be a sequent saturated with respect to both (trans) and (inter). If $\Delta$ is of form $\Delta',\langle\Sigma\Rightarrow\Pi\rangle$, then $\Gamma\Rightarrow\Delta\subseteq^\-S\Sigma\Rightarrow\Pi$.
\end{proposition}

In order to define a terminating  proof-search procedure based on \textbf{C}\calculus\ (like for any calculus with cumulative rules),  as usual we say that the backward application of a rule (R) to a sequent $S$ is {\em redundant} if $S$ satisfies the corresponding saturation condition for that application of (R) and we impose the following constraints: 

\indent (i)  {\em No rule is  applied to an axiom} and 

\indent (ii) {\em No rule is  applied redundantly.}

However the above restrictions are not sufficient to ensure the termination of the procedure as the following example shows.

\begin{example}\label{loop}
Let us consider the sequent $S = \Box a \supset \bot, \Box b \supset \bot  \Rightarrow p$, where we abbreviate by $\Gamma$ the antecedent of $S$. Consider the following derivation, we  only show the leftmost branch (the others succeed), we collapse some steps:
{
    \scriptsize
    \[
	\AxiomC{$\vdots$}
	\UnaryInfC{(3) $\Gamma \Rightarrow p, \Box a, \Box b, \langle \Gamma \Rightarrow \Box a, \Box b, [\Rightarrow a], \langle \Gamma \Rightarrow \Box a, \Box b, [\Rightarrow b]\rangle \rangle, \langle \Gamma \Rightarrow \Box a, \Box b, [\Rightarrow b]\rangle$}
	\UnaryInfC{$\vdots$}
	\UnaryInfC{(2) $\Gamma \Rightarrow p, \Box a, \Box b, \langle \Gamma \Rightarrow \Box a, \Box b, [\Rightarrow a], \langle \Rightarrow [\Rightarrow b]\rangle \rangle, \langle \Gamma \Rightarrow \Box a, \Box b, [\Rightarrow b]\rangle$}
	\RightLabel{$(\Box_R)$}
	\UnaryInfC{(1) $\Gamma \Rightarrow p, \Box a, \Box b, \langle \Gamma \Rightarrow \Box a, \Box b,  [\Rightarrow a]\rangle, \langle \Gamma \Rightarrow \Box a, \Box b,  [\Rightarrow b]\rangle$}
	\RightLabel{$(\supset_L)\times 4$}
	\UnaryInfC{$\Gamma \Rightarrow p, \Box a, \Box b, \langle \Gamma \Rightarrow [\Rightarrow a]\rangle, \langle \Gamma\Rightarrow [\Rightarrow b]\rangle$}
	\RightLabel{$(trans)\times 2$}
	\UnaryInfC{$\Gamma \Rightarrow p, \Box a, \Box b, \langle \Rightarrow [\Rightarrow a]\rangle, \langle \Rightarrow [\Rightarrow b]\rangle$}
	\RightLabel{$(\Box_R)\times 2$}
	\UnaryInfC{$\Gamma \Rightarrow p, \Box a, \Box b$}
	\RightLabel{$(\supset_L)\times 2$}
	\UnaryInfC{$\Gamma \Rightarrow p$}
	\DisplayProof
	\]
}
	
Observe that in sequent (1) $(\Box_R)$ can only be applied to $\Box b$, creating the  nested block $\langle \Rightarrow [\Rightarrow b]\rangle$ in (2), as it satisfies the saturation condition for $\Box a$. 
This block will be further expanded to $\langle \Gamma \Rightarrow \Box a, \Box b, [\Rightarrow b]\rangle$ in (3) that satisfies the saturation condition for $\Box b$, but not for $\Box a$, whence it will be further expanded, and so on. Thus the branch does not terminate. 

\end{example}
In order to deal with this situation, intuitively we need to block the expansion of a  sequent that occurs nested in another sequent whenever the former has already been expanded and the latter is "equivalent" to the former, in a sense that we will define. To accomplish this purpose we need to introduce a few notions. 

\begin{definition}[$\in^{\langle\cdot\rangle}, \in^{[\cdot]}, \in^+$-relation]
  Let $\Gamma_1\Rightarrow\Delta_1,\Gamma_2\Rightarrow\Delta_2$ be two sequents. We denote $\Gamma_1\Rightarrow\Delta_1 \in^{\langle\cdot\rangle}_0\Gamma_2\Rightarrow\Delta_2$ if $\langle \Gamma_1\Rightarrow\Delta_1 \rangle\in\Delta_2$. Let $\in^{\langle\cdot\rangle}$  be the transitive closure of $\in^{\langle\cdot\rangle}_0$. Relations $\in^{[\cdot]}_0$ and $\in^{[\cdot]}$ for modal blocks are  defined similarly. 
  Let $\in^+_0 =~\in^{\langle\cdot\rangle}_0 \cup \in^{[\cdot]}_0$ and finally let $\in^+$ be the reflexive-transitive closure of $\in^+_0$ . 
\end{definition}
Observe that $S'\in^+ S$ is the same as: for some context $G$, $S = G\{S'\}$.

\hide{
In order to trace the $\in^+$-components of a sequent more precisely, we assign an implication depth to each of them. 

\begin{definition}[Implication depth]
  Let $S$ be a sequent. For each $T\in^+S$, its implication depth $d_S(T)$ in $S$ is set as below.
  \begin{itemize}
    \item $d_S(S)=0$;
    \item $d_S(T)=d_S(T')$, if $T\in^{[\cdot]}_0T'$ and $T'\in^+S$;
    \item $d_S(T)=d_S(T')+1$, if $T\in^{\langle\cdot\rangle}_0T'$ and $T'\in^+S$.
  \end{itemize}
\end{definition}
}
We introduce the operator  $\sharp$ (to be compared with $*$ of Definition \ref{def:star-operator}). Its purpose is to remove implication blocks from a sequent and retain all other formulas.   

\begin{definition}[$\sharp$-operator]
  Let $\Lambda\Rightarrow\Theta$ be a sequent. We define  $\Theta^\sharp$ as follows: 
  (i) $\Theta^\sharp=\Theta$ if $\Theta$ is block-free; (ii)  $\Theta^\sharp=\Theta_0^\sharp,[\Phi\Rightarrow\Psi^\sharp]$ if $\Theta=\Theta_0,[\Phi\Rightarrow\Psi]$; (iii)
  $\Theta^\sharp=\Theta_0^\sharp$ if $\Theta=\Theta_0,\langle\Phi\Rightarrow\Psi\rangle$.
\end{definition}
As an example let $\Delta = b, [c\Rightarrow d, [e\Rightarrow f], \langle g \Rightarrow h \rangle],  \langle t \Rightarrow [p \Rightarrow q]\rangle, [m \Rightarrow n]$, then $\Delta^\sharp = b, [c\Rightarrow d, [e\Rightarrow f]], [m \Rightarrow n]$, while  $\Delta^* = [c\Rightarrow  [e\Rightarrow ]], [m \Rightarrow ]$.

Intuitively, if a sequent  $S= \Lambda \Rightarrow \Theta$ describes a model rooted in $S$ and specifies formulas forced and not forced in $S$, then   $\Lambda \Rightarrow \Theta^\sharp$, describes  the chains of R-related worlds to $S$  by specifying all formulas forced and not forced in each one of them, but ignores upper worlds in  the pre-order, the latter being represented by implication blocks.

We  use the $\sharp$-operator to define an equivalence relation  between sequents. The equivalence relation will be used to detect loops in a derivation as in the example above. 

\begin{definition}[Block-equivalence]\label{def:block-eq}
  Let $S_1, S_2$ be two sequents where $S_1=\Gamma_1\Rightarrow\Delta_1,S_2=\Gamma_2\Rightarrow\Delta_2$. We say $S_1$ is block-equivalent to $S_2$, denoted as $S_1\simeq S_2$, if $\Gamma_1=\Gamma_2$ and $\Delta_1^\sharp=\Delta_2^\sharp$.
\end{definition}

In order to define a proof-search procedure, we  divide rules of \textbf{C}\calculus~into three groups and define correspondingly  three levels of saturation.
\begin{enumerate}
  \item[(R1)] basic rules: all propositional and modal rules except $(\supset_R)$ and $(\square_R)$;
  \item[(R2)] rules that transfer formulas and blocks into implication blocks: (trans) and (inter);
  \item[(R3)] rules that create implication blocks: $(\square_R)$ and $(\supset_R)$.
\end{enumerate}

\begin{definition}[Saturation]\label{3-saturation}
  Let $S =\Gamma\Rightarrow\Delta$ be a sequent and not an axiom. $S$ is called: 
  \begin{itemize}
    \item {\rm R1-saturated} if $\Gamma\Rightarrow\Delta^\sharp$ satisfies all the saturation conditions  of  R1 rules;
    \item {\rm R2-saturated} if $S$ is R1-saturated and $S$ satisfies saturation conditions of R2 rules for blocks $S_1 \in^{\langle\cdot\rangle}_0 S$ and $S_2 \in^{[\cdot]}_0 S$.
    \item {\rm R3-saturated} if $S$ is R2-saturated and $S$ satisfies saturation conditions of  R3 rules for formulas $\square A,  B\supset C \in \Delta$.
  \end{itemize}
\end{definition}

We can finally define when a sequent is blocked, the intention is that it will not be expanded anymore by the proof-search procedure. 

\begin{definition}[Blocked sequent]
  Given a sequent $S$ and $S_1, S2 \in^+ S$, with   $S_1=\Gamma_1\Rightarrow\Delta_1,S_2=\Gamma_2\Rightarrow\Delta_2$. We say $S_2$ is blocked by $S_1$ in $S$, if $S_1$ is R3-saturated, $S_2\in^{\langle\cdot\rangle}S_1$ and $S_1\simeq S_2$. 
  We say that a sequent $S'$ is  {\rm blocked} in $S$ if there exists $S_1\in^+ S$ such that $S'$ is blocked by $S_1$ in $S$.
\end{definition}
Observe that if $S$ is finite, then for any $S'\in^+S$ checking whether $S'$ is  {\rm blocked} in $S$ can be effectively decided. We will say just that $S'$ is blocked when $S$ is clear.

\begin{example}
We reconsider the example \ref{loop}. The sequent (3) will be further expanded to 
{
  \scriptsize
  $$(4) \ \Gamma \Rightarrow p, \Box a, \Box b, \langle \Gamma \Rightarrow \Box a, \Box b, [\Rightarrow a], \langle \Gamma \Rightarrow \Box a, \Box b, [\Rightarrow b],\langle \Gamma \Rightarrow \Box a, \Box b, [\Rightarrow a]\rangle^{(ii)} \rangle \rangle^{(i)}, \langle \Gamma \Rightarrow \Box a, \Box b, [\Rightarrow b]\rangle$$ 
}
We have marked by (i) and (ii) the relevant blocks. Observe that the sequent $S_2 = \Gamma \Rightarrow \Box a, \Box b, [\Rightarrow a]$ in the block marked (ii)  is blocked by the sequent $S_1 = \Gamma \Rightarrow \Box a, \Box b, [\Rightarrow a], \langle \Gamma \Rightarrow \Box a, \Box b, [\Rightarrow b],\langle \Gamma \Rightarrow \Box a, \Box b, [\Rightarrow a]\rangle \rangle$ marked (i),  since $S_1$ is R3-saturated, $S_2\in^{\langle\cdot\rangle}S_1$ and $S_1\simeq S_2$, as in particular  $(\Box a, \Box b, [\Rightarrow a], \langle \Gamma \Rightarrow \Box a, \Box b, [\Rightarrow b],\langle \Gamma \Rightarrow \Box a, \Box b, [\Rightarrow a]\rangle \rangle )^\sharp = (\Gamma \Rightarrow \Box a, \Box b, [\Rightarrow a])^\sharp$.
\end{example}

We finally define three global saturation conditions.

\begin{definition}[Global saturation]
  Let $S$ be a sequent and not an axiom. $S$ is called :
  \begin{itemize}
  \item {\rm global-R1-saturated} if for each $T\in^+S$, $T$ is either R1-saturated or blocked;
    \item {\rm global-R2-saturated} if for each $T\in^+S$, $T$ is either R2-saturated or blocked;
    \item {\rm global-saturated} if for each $T\in^+S$, $T$ is either R3-saturated or blocked.
  \end{itemize}
\end{definition}

In order to specify the proof-search procedure, we make use of three sub-procedures that extend a given derivation $\+D$  by expanding a leaf $S$, each procedure applies rules {\em non-redundantly} to some $T:=\Gamma\Rightarrow\Delta\in^+S$, that we recall it means that $S = G\{T\}$, for some context $G$ . We define : 
\begin{enumerate}
  \item $\textbf{EXP1}(\+D, S, T)= \+D'$ where $\+D'$ is the extension of $\+D$  obtained by applying  R1 rules to every formula in $\Gamma\Rightarrow\Delta^\sharp$. 
  \item $\textbf{EXP2}(\+D, S, T)= \+D'$ where $\+D'$ is the extension of $\+D$  obtained by applying R2-rules  to blocks $\langle T_i\rangle,[T_j]\in\Delta$. 
  \item $\textbf{EXP3}(\+D, S, T)= \+D'$ where $\+D'$ is the extension of $\+D$  obtained by applying  R3-rules to formulas $\square A, A\supset B\in\Delta$. 
\end{enumerate}

The three procedures are used as macro-steps in the proof search procedure defined next. 

\begin{proposition}
Given a finite derivation $\+D$, a finite leaf $S$ of $\+D$ and $T\in^+S$,  then each  $\textbf{EXP1}(\+D, S, T)$, $\textbf{EXP2}(\+D, S, T)$,$\textbf{EXP3}(\+D, S, T)$  terminates by producing a finite  expansion of $\+D$ where all sequents in it are finite. 
\end{proposition}

Proof of this claim for $\textbf{EXP2}(\+D, S, T)$, $\textbf{EXP3}(\+D, S, T)$ is obvious, as only finitely many blocks or formulas in $T$ are processed. For $\textbf{EXP1}(\+D, S, T)$, the claim is less obvious,  since the rules are applied also deeply in $\Gamma\Rightarrow\Delta^\sharp$. However, notice that $\textbf{EXP1}$ only applies the rules (both L and R)  for  $\land, \lor, \Diamond$ and $\supset_L, \Box_L$ and ignores implication blocks, thus  $\textbf{EXP1}(\+D, S, T)$ produces exactly the same expansion of  $\+D$ that we would obtain  by the same rules of a nested sequent calculus for classical modal logic \textbf{K} \ \cite{brunner:2009}, and we know that it  terminates. 

Anyway, the claim for $\textbf{EXP1}(\+D, S, T)$ can be  proved by proving that any  derivation $\+Do$, with root  $\Gamma\Rightarrow\Delta^\sharp$ and generated by R1-rules, is finite.  Observe  that $\textbf{EXP1}(\+D, S, T)$ is obtained simply by  "appending" $\+Do$ to $\+D$,  where we   replace every sequent $T'$ in $\+Do$  by $G\{T'\}$, as  $S = G\{T\}$. In order to prove that $\+Do$ is finite, notice that (i) all R1-rules are at most binary, (ii) the length of a  branch of $\+Do$  is bounded by the size of the maximal sequent that  can occur in it  because of non-redundancy restriction. Thus we only need to estimate the size of any sequent in $\+Do$. In order to do so we introduce the following definition. 

\begin{definition}
  Given a sequent $S$, the tree $\+T_S$ is defined as follows:
  (i) the root of $\+T_S$ is $S$;
  (ii) if $S_1\in^{[\cdot]}_0S_2$, then $S_1$ is a child of $S_2$.
\end{definition}
We denote the height of $\+T_S$ as $h(T_S)$.  It is easy to verify that $h(\+T_S)\leq md(S)$. Moreover,  we have $|S| = \Sigma_{N\in \+T_S}|N|$, so that trivially $|S| \leq |N^{x}| \times Card(\+T_S)$, where $N^{x}$ is a a node of $\+T_S$ of maximal size. 
Moreover we denote  by $Sub(A)$  the set of subformulas of a formula $A$ and for a sequent $S= \Gamma \Rightarrow \Delta$ we use the corresponding notations  $Sub(\Gamma)$, $Sub(\Delta)$, $Sub(S)$. 
\hide{
Moreover, by $Sub(A)$ we denote the set of subformulas of a formula $A$, and 
 given a sequent $S=\Gamma\Rightarrow \Delta$, 
 where $\Delta = B_1,\ldots, B_k,[S_1],\ldots,[S_m],\langle T_1\rangle,\ldots,\langle T_n\rangle$,
we define   $Sub(\Delta) = \bigcup \{Sub(C) \mid C\in \Gamma\}  \cup  \bigcup \{ Sub(C) \mid C \in  \{B_1,\ldots, B_k\}\} \cup \bigcup_i Sub(Si) \cup \bigcup_j Sub(T_j)$ and then $Sub(S) = Sub(\Gamma) \cup Sub(\Delta)$.
}
Finally,  we recall  that $Card(Sub(S)) = O(|S|)$.

We get the following {\em rough} bound of the size of any sequent occurring in a derivation by R1-rules. 

\begin{proposition}\label{bound}
  Let $\+Do$ be a derivation with root a non-axiomatic sequent $T = \Gamma \Rightarrow \Delta$ obtained by applying R1-rules to  $\Gamma \Rightarrow \Delta^\sharp$, then any $T'$ occurring in $\+Do$  has size $O(|T|^{|T|+1})$. 
\end{proposition}

We  present below  the  proof-search procedure  PROCEDURE($A$), that given an input formula $A$ it returns either a proof of $A$ or a finite derivation tree in which all non-axiomatic leaves are global-saturated.

\begin{algorithm}[!htb]
  \KwIn{$\+D_0:=~\Rightarrow A$}
  initialization $\+D:=\+D_0$\;
    \Repeat{FALSE}{
      \uIf{all the leaves of $\+D$ are axiomatic}{
        return ``PROVABLE" and $\+D$ }
      \uElseIf{all the non-axiomatic leaves of $\+D$ are global-saturated}{return  ``UNPROVABLE" and $\+D$} 
    \Else{{\bf for} all non-axiomatic leaves $S$ of $\+D$ that are not global-saturated\\
         { \uIf{ $S$ is global-R2-saturated}{{\bf for} all $T\in^+S$ such that $T$ is a $\in^{\langle\cdot\rangle}$-minimal and not R3-saturated, check whether $T$ is blocked in $S$, if not, let  $\+D = \textbf{EXP3}(\+D, S, T)$}
        \uElseIf{ if  $S$ is global-R1-saturated}
            {{\bf for} all  $T\in^+ S$ that is not R2-saturated, let $\+D = \textbf{EXP2}(\+D, S, T)$}
         \uElse{{\bf for} all  $T\in^+ S$ that is not R1-saturated, let $\+D = \textbf{EXP1}(\+D, S, T)$}
          }
        }
    }
    \caption{PROCEDURE($A$)}\label{PROCEDURE}
\end{algorithm}

\hide{
Here we give some remarks for blocked sequent in this algorithm. As we can see, no rules are allowed to apply to a blocked sequent specified by $\textbf{CHECK}(\cdot)$, and once a sequent is marked as blocked, it cannot be activated any more, which means it will remain blocked even though the current sequent may be transformed into another one in the following steps. Besides, since $\textbf{CHECK}(\cdot)$ is applied to global-R2-saturated sequents, the blocked part of it is already R2-saturated.

Furthermore, we propose the following two lemmas in preparation for the termination.
}

An important property of the proof-search procedure is that saturation and blocking are preserved through sequent expansion, in other words they are {\em invariant} of the repeat loop of the procedure.

\begin{lemma}[Invariant]\label{lemma:invariant}
Let  $S$ be a leaf of a derivation $\+D$ with root $\Rightarrow A$: 
\begin{enumerate}
    \item Let $T\in^+ S$, where $T = \Gamma \Rightarrow \Delta$, for every rule (R) if $T$ satisfies the R-saturation condition on  some formulas $A_i$ and/or blocks $\langle T_j\rangle, [T_k]$ {\em before} the execution of (the body of) the repeat loop (lines 3-14), then $T$ satisfies the R-condition on the involved $A_i,\langle T_j\rangle, [T_k]$ {\em after} the execution of it.
    \item Let $T\in^+ S$, if $T$ is blocked in $S$ {\em before} the execution of (the body of) the repeat loop, then it is still so {\em after} it.
\end{enumerate}
\end{lemma}

The last ingredient in order to prove termination is that in a derivation of a formula $A$ there can be only finitely many non-blocked sequents.

\begin{lemma}\label{lem:eq-class-finite}
  Given  a formula $A$,  let $\textbf{Seq}(A)$ be the set of sequents that may occur in any possible derivation with root  $\Rightarrow A$. Let   $\textbf{Seq}(A)/_{\simeq}$ be the quotient of $\textbf{Seq}(A)$ with respect to block-equivalence $\simeq$ as defined in Definition \ref{def:block-eq}. Then $\textbf{Seq}(A)/_{\simeq}$ is finite.
\end{lemma}

Intuitively, the termination of the procedure is based on the following argument:   the procedure cannot run forever by building an infinite derivation. The reason is that the built derivation cannot  contain any infinite branch, because (i) once that a sequent satisfies a saturation condition for a rule R, further expansions of it will still satisfy that condition (whence not reconsidered for the application of R), (ii) if a sequent is blocked, further application or rules cannot "unblock" it, (iii) the number of non-equivalent, whence unblocked sequents is finite.

\begin{theorem}[Termination]\label{termination}
  Let $A$ be a formula. Proof-search for the sequent $\Rightarrow A$ terminates with a finite derivation in which any leaf is either an axiom or global-saturated.
\end{theorem}

Next, we prove the completeness of \textbf{C}\calculus. We show that  given a finite  global-saturated leaf $S$ of the derivation $\+D$  produced by PROCEDURE($A$), then we can define a countermodel $\+M_S$ for $A$ as follows:

\begin{definition}
  The model $\+M_S=(W_S,\leq_S,R_S,V_S)$ determined by $S$ is defined as follows:
  \begin{itemize}
    \item $W_S=\{x_{\Phi\Rightarrow\Psi}~|~\Phi\Rightarrow\Psi\in^+S\}$.
    \item the relation $\leq_S$, for $x_{S_1},x_{S_2}\in W_S$ is defined by $x_{S_1}\leq_S x_{S_2}$ if $S_1\subseteq^\-S S_2$.
    \item The accessibility relation $R_S$, for $x_{S_1},x_{S_2}\in W_S$, is defined by $R_S x_{S_1}x_{S_2}$ if $S_2\in^{[\cdot]}_0S_1$.
    \item For the valuation $V_S$, for each $x_{\Phi\Rightarrow\Psi}\in W_S$, let $V_S(x_{\Phi\Rightarrow\Psi})=\{p~|~p\in \Phi\}$.
  \end{itemize}
\end{definition}
Obviously $\+M_S$  is finite; each world in $W_S$ corresponds to either a R3-saturated or a blocked sequent, that is nonetheless saturated with respect to (inter) and (trans).  Moreover, 
if  $x_{\Gamma\Rightarrow\Delta',\langle\Sigma\Rightarrow\Pi\rangle} \in W_S$ then $x_{\Sigma\Rightarrow\Pi}\in W_S$, and   $x_{\Gamma\Rightarrow\Delta',\langle\Sigma\Rightarrow\Pi\rangle} \leq_S x_{\Sigma\Rightarrow\Pi}$.  By the property of  structural inclusion $\subseteq^\-S$, we have that $\leq_S$  is a pre-order.

\begin{proposition}\label{hp-fc-property}
   $\+M_S$ satisfies the hereditary property (HP) and forward confluence (FC).
\end{proposition}

\begin{lemma}[Truth Lemma]\label{truth-lemma-formula}
  Let $S$ be a global-saturated sequent and $\+M_{S}$ be defined as above. (a). If $A\in\Phi$, then $\+M_{S},x_{\Phi\Rightarrow\Psi}\Vdash A$; (b). If $A\in\Psi$, then $M_{S},x_{\Phi\Rightarrow\Psi}\nVdash A$.
\end{lemma}

From the truth lemma we immediately obtain the completeness of \textbf{C}\calculus.  

\begin{theorem}\label{completeness-cfik}
  For any formula $A\in \+L$, if $\Vdash A$, then $\Rightarrow A$ is provable in {\rm \textbf{C}\calculus}.
\end{theorem}

\begin{example}
	We show how to build a countermodel of the formula $(\Diamond p \supset \Box q) \supset \Box (p \supset q)$ by \textbf{C}\calculus \ (because of space limit, we omit the steps of the derivation). Ignoring the first step, a  derivation is initialized with $\Diamond p \supset \Box q \Rightarrow \Box (p \supset q)$. By  backward application of rules, one branch of the derivation ends up with the the saturated sequent $S_0$ : 
	$$S_0 = \ \Diamond p \supset \Box q \Rightarrow \Diamond p,  \Box (p \supset q), \langle \Diamond p \supset \Box q \Rightarrow \Diamond p,  [\Rightarrow p \supset q, \langle p \Rightarrow q\rangle, p]\rangle \quad \mbox{and let:}$$
	\begin{quote}$S_1 = \Diamond p \supset \Box q \Rightarrow \Diamond p,  [\Rightarrow p \supset q, \langle p \Rightarrow q \rangle, p]$ \quad
		$S_2 = \Rightarrow p \supset q, \langle p \Rightarrow q\rangle, p$ \quad 
		$S_3 = p \Rightarrow q$
	\end{quote}
	We then get the model $M_{S_0} = (W, \leq, R, V)$ where $W=\{x_{S_0}, x_{S_1}, x_{S_2}, x_{S_3} \}$ $x_{S_0} \leq x_{S_1}$, $x_{S_2} \leq x_{S_0}$, $x_{S_2} \leq x_{S_3}$, $R x_{S_1} x_{S_2}$, and $V(x_{S_0}) = V(x_{S_1}) = V(x_{S_2}) =\emptyset$ and 
$V(x_{S_3}) = \{p\}$. It is easy to see that $x_{S_0} \not\Vdash (\Diamond p \supset \Box q) \supset \Box (p \supset q)$.	
\end{example}

\begin{example}\label{exampleBoxIk}
This example shows that the $\Diamond$-free fragment of  \textbf{FIK} is weaker than the same fragment of \textbf{IK}. Let us consider the  formula
$\neg\neg \Box \neg p\supset \Box\neg p$ presented in \cite{anupam:2023},  which is provable in \textbf{IK}. On the other hand if we build a derivation with root $\Rightarrow ((\Box (p \supset \bot)\supset \bot)\supset\bot) \supset \Box (p \supset \bot)$, we generate the  saturated sequent $S_0 = F \Rightarrow \Box (p \supset \bot), G, \langle S_1\rangle, \langle S_6\rangle$, where $F= (\Box (p \supset \bot)\supset \bot)\supset\bot$ and $G=  \Box (p \supset \bot)\supset \bot$, and  \\
$S_1 =  F \Rightarrow G, [\Rightarrow \langle p \Rightarrow \bot\rangle ], \langle S_4\rangle$, \quad $S_4 = F, \Box (p \supset \bot) \Rightarrow \bot, G, [p \supset\bot \Rightarrow p]$, \\
$S_6 =  F, \Box (p \supset \bot) \Rightarrow \bot, G$. \\
Further let $S_2 = \Rightarrow \langle p \Rightarrow \bot\rangle$, $S_3 = p \Rightarrow \bot$, $S_5 =  p \supset\bot \Rightarrow p$.

We get the model $M_{S_0} = (W, \leq, R, V)$ where $W=\{x_{S_0},\ldots,x_{S_6} \}$, $x_{S_0} \leq x_{S_1}, x_{S_0} \leq x_{S_6},  x_{S_1} \leq x_{S_4},  x_{S_6} \leq x_{S_4}$, $x_{S_2} \leq x_{S_3}$,  $x_{S_2} \leq x_{S_5}$  $x_{S_2} \leq x_{S_0}$ , $R x_{S_1} x_{S_2}$,  $R x_{S_4} x_{S_5}$,  $V(x_{S_i})  =\emptyset$ for $i\not = 3$ and	$V(x_{S_3}) = \{p\}$. It is easy to see that $x_{S_0} \not\Vdash \Box (p \supset \bot)$, as  $x_{S_0} \leq x_{S_1} R  x_{S_3}$ and $ x_{S_3} \Vdash p$; moreover $x_{S_0} \Vdash F$ since $x_{S_5} \Vdash p\supset \bot$, whence $x_{S_4} \Vdash \Box (p \supset \bot)$ and  $\forall y \geq x_{S_0}. y \leq x_{S_4}$. Observe that $M$ satisfies (FC), the only worlds which are concerned are  $x_{S_1}, x_{S_2}, x_{S_4}, x_{S_5}$.
\end{example}

\section{Conclusion and future work}
We have proposed \logic,  a natural variant of Intuitionistic modal logic characterized by forward confluent bi-relational models.
\logic~is intermediate between Constructive Modal logic \textbf{CK} and Intuitionistic Modal Logic \textbf{IK} and it satisfies all the expected criteria for \textbf{IML}.
We have presented a sound and complete axiomatization of it and  a bi-nested calculus \calculus~which provides a decision procedure together with a finite countermodel extraction. 

There are many topics for further research. First we may study extensions of \logic~with the standard axioms from the modal cube. Moreover we can consider other bi-relational frame conditions relating the pre-order and the accessible  (including the one for \textbf{IK}) and see how they can be captured uniformly in Bi-nested calculi with suitable "interaction rules". 

\section*{Acknowledgement}

This paper is originated from a discussion started by Anupam Das and Sonia Marin in the proof theory blog (see the link \url{https://prooftheory.blog/2022/08/19/}), we are grateful to them, as well as to all other contributors to the discussion. In particular Example \ref{examplefik} was reported in the  blog by Alex Simpson, who had learnt it  in 1996 by Carsten Grefe in private communication. Example \ref{exampleBoxIk} was suggested first by Anupam Das and Sonia Marin in the blog. Special thanks to Marianna Girlando for fruitful discussions.

\bibliography{}
\bibliographystyle{plainurl}

\appendix

\section*{Appendix}

This Appendix includes the proofs of some of our results. 
Some of these proofs are relatively simple and we have included them here just for the sake of the completeness.

~

{\noindent \sf \textbf{Lemma~\ref{fc:canonical:frame:is:forward:confluent}.}} 
$(W_{c},\leq_{c},R_{c},V_{c})$ satisfies the frame condition {\bf (FC)}.

\begin{proof}
Let $\Gamma,\Delta,\Lambda\in W_{c}$ be such that $\Gamma\geq_{c}\Delta$ and $\Delta R_{c}\Lambda$. Hence, $\Gamma\supseteq\Delta$ and $\Delta\bowtie\Lambda$. Let $A_{1},A_{2},\ldots$ be an enumeration of $\square\Gamma$ and $B_{1},B_{2},\ldots$ be an enumeration of $\Lambda$. Obviously, for all $n\in\N$, $\square(A_{1}\wedge\ldots\wedge A_{n})\in\Gamma$ and $B_{1}\wedge\ldots\wedge B_{n}\in\Lambda$. Since $\Delta\bowtie\Lambda$, for all $n\in\N$, $\Diamond(B_{1}\wedge\ldots\wedge B_{n})\in\Delta$. For all $n\in\N$, let $\Theta_{n}=$\Derivable$+A_{1}\wedge\ldots\wedge A_{n}\wedge B_{1}\wedge\ldots\wedge B_{n}$. Obviously, $(\Theta_{n})_{n\in\N}$ is a chain of theories such that $\bigcup\{\Theta_{n}:\ n\in\N\}\supseteq\Lambda$.

We claim that for all formulas $C$, if $\square C\in\Gamma$ then $C\in\bigcup\{\Theta_{n}:\ n\in\N\}$.
If not, there exists a formula $C$ such that $\square C\in\Gamma$ and $C\not\in\bigcup\{\Theta_{n}:\ n\in\N\}$.
Thus, $C\in\square\Gamma$.
Consequently, let $n\in\N$ be such that $A_{n}=C$.
Hence, $A_{1}\wedge\ldots\wedge A_{n}\wedge B_{1}\wedge\ldots\wedge B_{n}\rightarrow C$ is in \Derivable.
Thus, $C\in\Theta_{n}$.
Consequently, $C\in\bigcup\{\Theta_{n}:\ n\in\N\}$: a contradiction.
Hence, for all formulas $C$, if $\square C\in\Gamma$ then $C\in\bigcup\{\Theta_{n}:\ n\in\N\}$.

We claim that for all formulas $C$, if $C\in\bigcup\{\Theta_{n}:\ n\in\N\}$ then $\Diamond C\in\Gamma$.
If not, there exists $n\in\N$ and there exists a formula $C$ such that $C\in\Theta_{n}$ and $\Diamond C\not\in\Gamma$.
Thus, $A_{1}\wedge\ldots\wedge A_{n}\wedge B_{1}\wedge\ldots\wedge B_{n}\rightarrow C$ is in \Derivable.
Consequently, $B_{1}\wedge\ldots\wedge B_{n}\rightarrow(A_{1}\wedge\ldots\wedge A_{n}\rightarrow C)$ is in \Derivable.
Hence, $\Diamond(B_{1}\wedge\ldots\wedge B_{n})\supset\Diamond(A_{1}\wedge\ldots\wedge A_{n}\supset C)$ is in \Derivable.
Since $\Diamond(B_{1}\wedge\ldots\wedge B_{n})\in\Delta$, $\Diamond(A_{1}\wedge\ldots\wedge A_{n}\supset C)\in\Delta$.
Since $\Gamma\supseteq\Delta$, $\Diamond(A_{1}\wedge\ldots\wedge A_{n}\supset C)\in\Gamma$.
Thus, $\square(A_{1}\wedge\ldots\wedge A_{n})\supset\Diamond C\in\Gamma$.
Since $\square(A_{1}\wedge\ldots\wedge A_{n})\in\Gamma$, $\Diamond C\in\Gamma$: a contradiction.
Consequently, for all formulas $C$, if $C\in\bigcup\{\Theta_{n}:\ n\in\N\}$ then $\Diamond C\in\Gamma$.

Let ${\mathcal S}=\{\Theta:\ \Theta$ is a theory such that {\bf (1)}~$\Gamma\bowtie\Theta$ and {\bf (2)}~$\Theta\supseteq\Lambda\}$. Obviously, $\bigcup\{\Theta_{n}:\ n\in\N\}\in{\mathcal S}$.
Hence, ${\mathcal S}$ is nonempty.
Moreover, for all nonempty chains $(\Pi_{i})_{i\in I}$ of elements of ${\mathcal S}$, $\bigcup\{\Pi_{i}:\ i\in I\}$ is an element of ${\mathcal S}$.
Thus, by Zorn's Lemma, ${\mathcal S}$ possesses a maximal element $\Theta$.
Consequently, $\Theta$ is a theory such that $\Gamma\bowtie\Theta$ and $\Theta\supseteq\Lambda$. 
Hence, it only remains to be proved that $\Theta$ is proper and prime.

We claim that $\Theta$ is proper.
If not, $\bot\in\Theta$.
Since $\Gamma\bowtie\Theta$, $\Diamond\bot\in\Gamma$: a contradiction.
Thus, $\Theta$ is proper.

We claim that $\Theta$ is prime.
If not, there exists formulas $C,D$ such that $C\vee D\in\Theta$, $C\not\in\Theta$ and $D\not\in\Theta$.
Consequently, by the maximality of $\Theta$ in ${\mathcal S}$, $\Theta+C\not\in{\mathcal S}$ and $\Theta+D\not\in{\mathcal S}$.
Hence, there exists a formula $E$ such that $E\in\Theta+C$ and $\Diamond E\not\in\Gamma$ and there exists a formula $F$ such that $F\in\Theta+D$ and $\Diamond F\not\in\Gamma$.
Thus, $C\supset E\in\Theta$ and $D\supset F\in\Theta$.
Consequently, $C\vee D\supset E\vee F\in\Theta$.
Since $C\vee D\in\Theta$, $E\vee F\in\Theta$.
Since $\Gamma\bowtie\Theta$, $\Diamond(E\vee F)\in\Gamma$.
Hence, either $\Diamond E\in\Gamma$, or $\Diamond F\in\Gamma$: a contradiction.
Thus, $\Theta$ is prime.
\end{proof}

{\noindent \sf \textbf{Lemma~\ref{lemma:prime:proper:for:implication}.}} 
Let $\Gamma$ be a prime theory. Let $B,C$ be formulas.
\begin{enumerate}
  \item If $B\supset C\not\in\Gamma$ then there exists a prime theory $\Delta$ such that $\Gamma\subseteq\Delta$, $B\in\Delta$ and $C\not\in\Delta$,
  \item if $\square B\not\in\Gamma$ then there exists prime theories $\Delta,\Lambda$ such that $\Gamma\subseteq\Delta$, $\Delta\bowtie\Lambda$ and $B\not\in\Lambda$,
  \item if $\Diamond B\in\Gamma$ then there exists a prime theory $\Delta$ such that $\Gamma\bowtie\Delta$ and $B\in\Delta$.
\end{enumerate}

\begin{proof}
{\bf Case of $\supset$.}
Suppose $B\supset C\not\in\Gamma$.
Let ${\mathcal S}=\{\Delta:\ \Delta$ is a theory such that {\bf (1)}~$\Gamma\subseteq\Delta$, {\bf (2)}~$B\in\Delta$ and {\bf (3)}~$C\not\in\Delta\}$.

Since $B\supset C\not\in\Gamma$, $C\not\in\Gamma+B$.
Hence, $\Gamma+B\in{\mathcal S}$.
Thus, ${\mathcal S}$ is nonempty.
Moreover, for all nonempty chains $(\Delta_{i})_{i\in I}$ of elements of ${\mathcal S}$, $\bigcup\{\Delta_{i}:\ i\in I\}$ is an element of ${\mathcal S}$.
Consequently, by Zorn's Lemma, ${\mathcal S}$ possesses a maximal element $\Delta$.
Hence, $\Delta$ is a theory such that $\Gamma\subseteq\Delta$, $B\in\Delta$ and $C\not\in\Delta$.

Thus, it only remains to be proved that $\Delta$ is proper and prime.

We claim that $\Delta$ is proper.
If not, $\Delta=\+L$.
Consequently, $C\in\Delta$: a contradiction.
Hence, $\Delta$ is proper.

We claim that $\Delta$ is prime.
If not, there exists formulas $D,E$ such that $D\vee E\in\Delta$, $D\not\in\Delta$ and $E\not\in\Delta$.
Thus, by the maximality of $\Delta$ in ${\mathcal S}$, $\Delta+D\not\in{\mathcal S}$ and $\Delta+E\not\in{\mathcal S}$.
Consequently, $C\in\Delta+D$ and $C\in\Delta+E$.
Hence, $D\supset C\in\Delta$ and $E\supset C\in\Delta$.
Thus, $D\vee E\supset C\in\Delta$.
Since $D\vee E\in\Delta$, $C\in\Delta$: a contradiction.
Consequently, $\Delta$ is prime.

{\bf Case of $\square$.}
Suppose $\square B\not\in\Gamma$.
Let ${\mathcal S}=\{\Delta:\ \Delta$ is a theory such that {\bf (1)}~$\Gamma\subseteq\Delta$ and {\bf (2)}~$\square B\not\in\Delta\}$.

Since $\square B\not\in\Gamma$, $\Gamma\in{\mathcal S}$.
Hence, ${\mathcal S}$ is nonempty.
Moreover, for all nonempty chains $(\Delta_{i})_{i\in I}$ of elements of ${\mathcal S}$, $\bigcup\{\Delta_{i}:\ i\in I\}$ is an element of ${\mathcal S}$.
Thus, by Zorn's Lemma, ${\mathcal S}$ possesses a maximal element $\Delta$.
Consequently, $\Delta$ is a theory such that $\Gamma\subseteq\Delta$ and $\square B\not\in\Delta$.

We claim that $\Delta$ is proper.
If not, $\Delta=\+L$.
Hence, $\square B\in\Delta$: a contradiction.
Thus, $\Delta$ is proper.

We claim that $\Delta$ is prime.
If not, there exists formulas $C,D$ such that $C\vee D\in\Delta$, $C\not\in\Delta$ and $D\not\in\Delta$.
Consequently, by the maximality of $\Delta$ in ${\mathcal S}$, $\Delta+C\not\in{\mathcal S}$ and $\Delta+D\not\in{\mathcal S}$.
Hence, $\square B\in\Delta+C$ and $\square B\in\Delta+D$.
Thus, $C\supset\square B\in\Delta$ and $D\supset\square B\in\Delta$.
Consequently, $C\vee D\supset\square B\in\Delta$.
Since $C\vee D\in\Delta$, $\square B\in\Delta$: a contradiction.
Hence, $\Delta$ is prime.

We claim that for all formulas $C$, if $C\vee B\in\square\Delta$ then $\Diamond C\in\Delta$.
If not, there exists a formula $C$ such that $C\vee B\in\square\Delta$ and $\Diamond C\not\in\Delta$.
Thus, by the maximality of $\Delta$ in ${\mathcal S}$, $\Delta+\Diamond C\not\in{\mathcal S}$.
Consequently, $\square B\in\Delta+\Diamond C$.
Hence, $\Diamond C\supset\square B\in\Delta$.
Since $C\vee B\in\square\Delta$, $\square(C\vee B)\in\Delta$.
Since $\Diamond C\supset\square B\in\Delta$, $\square B\in\Delta$: a contradiction.
Thus, for all formulas $C$, if $C\vee B\in\square\Delta$ then $\Diamond C\in\Delta$.

Let ${\mathcal T}=\{\Lambda:\ \Lambda$ is a theory such that {\bf (1)}~$\square\Delta\subseteq\Lambda$, {\bf (2)}~for all formulas $C$, if $C\vee B\in\Lambda$ then $\Diamond C\in\Delta$ and {\bf (3)}~$B\not\in\Lambda\}$.

Since $\square B\not\in\Delta$, $B\not\in\square\Delta$.
Consequently, $\square\Delta\in{\mathcal T}$.
Hence, ${\mathcal T}$ is nonempty.
Moreover, for all nonempty chains $(\Lambda_{i})_{i\in I}$ of elements of ${\mathcal T}$, $\bigcup\{\Lambda_{i}:\ i\in I\}$ is an element of ${\mathcal T}$.
Thus, by Zorn's Lemma, ${\mathcal T}$ possesses a maximal element $\Lambda$.
Consequently, $\Lambda$ is a theory such that $\square\Delta\subseteq\Lambda$, for all formulas $C$, if $C\vee B\in\Lambda$ then $\Diamond C\in\Delta$ and $B\not\in\Lambda$.

Hence, it only remains to be proved that $\Lambda$ is proper and prime and $\Delta\bowtie\Lambda$.

We claim that $\Lambda$ is proper. If not, $\Lambda=\+L$. Thus, $B\in\Lambda$: a contradiction. Consequently, $\Lambda$ is proper.

We claim that $\Lambda$ is prime.
If not, there exists formulas $C,D$ such that $C\vee D\in\Lambda$, $C\not\in\Lambda$ and $D\not\in\Lambda$.
Hence, by the maximality of $\Lambda$ in ${\mathcal T}$, $\Lambda+C\not\in{\mathcal T}$ and $\Lambda+D\not\in{\mathcal T}$.
Thus, either there exists a formula $E$ such that $E\vee B\in\Lambda+C$ and $\Diamond E\not\in\Delta$, or $B\in\Lambda+C$ and either there exists a formula $F$ such that $F\vee E\in\Lambda+D$ and $\Diamond F\not\in\Delta$, or $B\in\Lambda+D$.
Consequently, we have to consider the following 4 cases.

$\mathbf{(1)}$ Case ``there exists a formula $E$ such that $E\vee B\in\Lambda+C$ and $\Diamond E\not\in\Delta$ and there exists a formula $F$ such that $F\vee B\in\Lambda+D$ and $\Diamond F\not\in\Delta$'':
Hence, $C\supset E\vee B\in\Lambda$ and $D\supset F\vee B\in\Lambda$.
Thus, $C\vee D\supset E\vee F\vee B\in\Lambda$.
Since $C\vee D\in\Lambda$, $E\vee F\vee B\in\Lambda$.
Consequently, $\Diamond(E\vee F)\in\Delta$.
Hence, either $\Diamond E\in\Delta$, or $\Diamond F\in\Delta$: a contradiction.

$\mathbf{(2)}$ Case ``there exists a formula $E$ such that $E\vee F\in\Lambda+C$ and $\Diamond E\not\in\Delta$ and $B\in\Lambda+D$'':
Thus, $C\supset E\vee B\in\Lambda$ and $D\supset B\in\Lambda$.
Consequently, $C\vee D\supset E\vee B\in\Lambda$.
Since $C\vee D\in\Lambda$, $E\vee B\in\Lambda$.
Hence, $\Diamond E\in\Delta$: a contradiction.

$\mathbf{(3)}$ Case ``$B\in\Lambda+C$ and there exists a formula $F$ such that $F\vee B\in\Lambda+D$ and $\Diamond F\not\in\Delta$'':
Thus, $C\supset B\in\Lambda$ and $D\supset F\vee B\in\Lambda$.
Consequently, $C\vee D\supset F\vee B\in\Lambda$.
Since $C\vee D\in\Lambda$, $F\vee B\in\Lambda$.
Hence, $\Diamond F\in\Delta$: a contradiction.

$\mathbf{(4)}$ Case ``$B\in\Lambda+C$ and $B\in\Lambda+D$'':
Thus, $C\supset B\in\Lambda$ and $D\supset B\in\Lambda$.
Consequently, $C\vee D\supset B\in\Lambda$.
Since $C\vee D\in\Lambda$, $B\in\Lambda$: a contradiction. 

Hence, $\Lambda$ is prime.

We claim that $\Delta\bowtie\Lambda$.
If not, there exists a formula $C$ such that $C\in\Lambda$ and $\Diamond C\not\in\Delta$.
Thus, $C\vee B\in\Lambda$.
Consequently $\Diamond C\in\Delta$: a contradiction.
Hence, $\Delta\bowtie\Lambda$.

{\bf Case of $\Diamond$.} 
Suppose $\Diamond B\in\Gamma$. Let ${\mathcal S}=\{\Delta:\ \Delta$ is a theory such that {\bf (1)}~for all formulas $C$, if $C\in\Delta$ then $\Diamond C\in\Gamma$ and {\bf (2)}~$B\in\Delta\}$.

We claim that $\square\Gamma+B\in{\mathcal S}$.
If not, there exists a formula $C$ such that $C\in\square\Gamma+B$ and $\Diamond C\not\in\Gamma$.
Hence, $B\supset C\in\square\Gamma$.
Thus, $\square(B\supset C)\in\Gamma$.
Consequently, $\Diamond B\supset\Diamond C\in\Gamma$.
Since $\Diamond B\in\Gamma$, $\Diamond C\in\Gamma$: a contradiction.
Hence, $\square\Gamma+B\in{\mathcal S}$.
Thus, ${\mathcal S}$ is nonempty.
Moreover, for all nonempty chains $(\Delta_{i})_{i\in I}$ of elements of ${\mathcal S}$, $\bigcup\{\Delta_{i}:\ i\in I\}$ is an element of ${\mathcal S}$.
Consequently, by Zorn's Lemma, ${\mathcal S}$ possesses a maximal element $\Delta$.
Hence, $\Delta$ is a theory such that for all formulas $C$, if $C\in\Delta$ then $\Diamond C\in\Gamma$ and $B\in\Delta$.

Thus, it only remains to be proved that $\Delta$ is proper and prime and $\Gamma\bowtie\Delta$.

We claim that $\Delta$ is proper.
If not, $\bot\in\Delta$.
Consequently, $\Diamond\bot\in\Gamma$: a contradiction.
Hence, $\Delta$ is proper.

We claim that $\Delta$ is prime.
If not, there exists formulas $C,D$ such that $C\vee D\in\Delta$, $C\not\in\Delta$ and $D\not\in\Delta$.
Thus, by the maximality of $\Delta$ in ${\mathcal S}$, $\Delta+C\not\in{\mathcal S}$ and $\Delta+D\not\in{\mathcal S}$.
Consequently, there exists a formula $E$ such that $E\in\Delta+C$ and $\Diamond E\not\in\Gamma$ and there exists a formula $F$ such that $F\in\Delta+D$ and $\Diamond F\not\in\Gamma$.
Hence, $C\supset E\in\Delta$ and $D\supset F\in\Delta$.
Thus, $C\vee D\supset E\vee F\in\Delta$.
Since $C\vee D\in\Delta$, $E\vee F\in\Delta$.
Consequently, $\Diamond(E\vee F)\in\Gamma$.
Hence, either $\Diamond E\in\Gamma$, or $\Diamond F\in\Gamma$: a contradiction.
Thus, $\Delta$ is prime.

We claim that $\Gamma\bowtie\Delta$.
If not, there exists a formula $C$ such that $\square C\in\Gamma$ and $C\not\in\Delta$.
Consequently, by the maximality of $\Delta$ in ${\mathcal S}$, $\Delta+C\not\in{\mathcal S}$.
Hence, there exists a formula $D$ such that $D\in\Delta+C$ and $\Diamond D\not\in\Gamma$.
Thus, $C\supset D\in\Delta$.
Consequently, $\Diamond(C\supset D)\in\Gamma$.
Since $\square C\in\Gamma$, $\Diamond D\in\Gamma$: a contradiction.
Hence, $\Gamma\bowtie\Delta$.
\end{proof}

{\noindent \sf \textbf{Lemma~\ref{lemma:truth:lemma}.}} 
For all formulas $A$ and for all $\Gamma\in W_{c}$, $A\in\Gamma$ if and only if $\Gamma\models A$.

\begin{proof}
By induction on $A$. We only consider the following 3 cases.

$\mathbf{(1)}$ Case ``there exists formulas $B,C$ such that $A=B\supset C$'':
Let $\Gamma\in W_{c}$.
From left to right, suppose $B\supset C\in\Gamma$ and $\Gamma\not\models B\supset C$.
Hence, there exists $\Delta\in W_{c}$ such that $\Gamma\leq_{c}\Delta$, $\Delta\models B$ and $\Delta\not\models C$.
Thus, $\Gamma\subseteq\Delta$.
Moreover, by induction hypothesis, $B\in\Delta$ and $C\not\in\Delta$.
Since $B\supset C\in\Gamma$, $B\supset C\in\Delta$.
Since $B\in\Delta$, $C\in\Delta$: a contradiction.
From right to left, suppose $\Gamma\models B\supset C$ and $B\supset C\not\in\Gamma$.
Consequently, by Lemma~\ref{lemma:prime:proper:for:implication}, there exists a prime theory $\Delta$ such that $\Gamma\subseteq\Delta$, $B\in\Delta$ and $C\not\in\Delta$.
Hence, $\Gamma\leq_{c}\Delta$.
Moreover, by induction hypothesis, $\Delta\models B$ and $\Delta\not\models C$.
Thus, $\Gamma\not\models B\supset C$: a contradiction.

$\mathbf{(2)}$ Case ``there exists a formula $B$ such that $A=\square B$'':
Let $\Gamma\in W_{c}$.
From left to right, suppose $\square B\in\Gamma$ and $\Gamma\not\models\square B$.
Thus, there exists $\Delta,\Lambda\in W_{c}$ such that $\Gamma\leq_{c}\Delta$, $\Delta R_{c}\Lambda$ and $\Lambda\not\models B$.
Consequently, $\Gamma\subseteq\Delta$ and $\Delta\bowtie\Lambda$.
Moreover, by induction hypothesis, $B\not\in\Lambda$.
Since $\square B\in\Gamma$, $B\in\Lambda$: a contradiction.
From right to left, suppose $\Gamma\models\square B$ and $\square B\not\in\Gamma$.
Hence, by Lemma~\ref{lemma:prime:proper:for:implication}, there exists prime theories $\Delta,\Lambda$ such that $\Gamma\subseteq\Delta$, $\Delta\bowtie\Lambda$ and $B\not\in\Lambda$.
Thus, $\Gamma\leq_{c}\Delta$ and $\Delta R_{c}\Lambda$.
Moreover, by induction hypothesis, $\Lambda\not\models B$.
Consequently, $\Gamma\not\models\square B$: a contradiction.

$\mathbf{(3)}$ Case ``there exists a formula $B$ such that $A=\Diamond B$'':
Let $\Gamma\in W_{c}$.
From left to right, suppose $\Diamond B\in\Gamma$ and $\Gamma\not\models\Diamond B$.
Consequently, by Lemma~\ref{lemma:prime:proper:for:implication}, there exists a prime theory $\Delta$ such that $\Gamma\bowtie\Lambda$ and $B\in\Delta$.
Hence, $\Gamma R_{c}\Delta$.
Moreover, by induction hypothesis, $\Delta\models B$.
Thus, $\Gamma\models\Diamond B$: a contradiction.
From right to left, suppose $\Gamma\models\Diamond B$ and $\Diamond B\not\in\Gamma$.
Consequently, there exists $\Delta\in W_{c}$ such that $\Gamma R_{c}\Delta$ and $\Delta\models B$.
Hence, $\Gamma\bowtie\Delta$.
Moreover, by induction hypothesis, $B\in\Delta$.
Since $\Diamond B\not\in\Gamma$, $B\not\in\Delta$: a contradiction.
\end{proof}

{\noindent \sf \textbf{Lemma~\ref{lemma:about:excluded:middle:consequences}.}} 
$\Diamond p\equiv\neg\square\neg p$ is in \Derivable$^{+}$.

\begin{proof}
$(1)$~Obviously, $\neg p\supset(p\supset\bot)$ is in \Derivable$^{+}$.
Hence, using $(\NEC)$ and $(\K_{\square})$, $\square\neg p\supset\square(p\supset\bot)$ is in \Derivable$^{+}$.
Thus, $\Diamond p\supset\neg\square\neg p\vee\square(p\supset\bot)$ is in \Derivable$^{+}$.
Consequently, using $(\K_{\Diamond})$, $\Diamond p\supset\neg\square\neg p\vee\Diamond\bot$ is in \Derivable$^{+}$.
Since $\neg\Diamond\bot$ is in \Derivable$^{+}$, $\Diamond p\supset\neg\square\neg p$ is in \Derivable$^{+}$.

$(2)$~Obviously, $p\vee\neg p$ is in \Derivable$^{+}$.
Hence, using $(\NEC)$, $\square(p\vee\neg p)$ is in \Derivable$^{+}$.
Since using $(\wCD)$, $\square(p\vee\neg p)\supset((\Diamond p\supset\square\neg p)\supset\square\neg p)$ is in \Derivable$^{+}$, $(\Diamond p\supset\square\neg p)\supset\square\neg p$ is in \Derivable$^{+}$.
Thus, $\neg\square\neg p\supset\Diamond p$ is in \Derivable$^{+}$.
\end{proof}

{\noindent \sf \textbf{Lemma~\ref{lemma:about:feature:5}.}} 
Let $p$ be an atomic proposition. There exists no $\square$-free $A$ such that $\square p\equiv A$ is in \Derivable\ and there exists no $\Diamond$-free $A$ such that $\Diamond p\equiv A$ is in \Derivable.

\begin{proof}
$(1)$~For the sake of the contradiction, suppose there exists a $\square$-free formula $A$ such that $\square p\equiv A$ is in \Derivable.
Without loss of generality, we may assume that $p$ is the only atomic proposition that may occur in $A$.
Since $\square p\equiv A$ is in \Derivable, by Theorem~\ref{theorem:soundness:dfik}, $\Vdash\square p\equiv A$.
Let $(W,\leq,R,V)$ be the bi-relational model defined by $W=\{a,b,c,d\}$, $a\leq c$, $b\leq d$, $aRb$, $aRd$, $cRd$ and $V(p)=\{d\}$.
By induction on the $\square$-free formula $B$, the reader may easily verify that $\+M,a\Vdash B$ if and only if $\+M,c\Vdash B$.
Since $A$ is $\square$-free, $\+M,a\Vdash A$ if and only if $\+M,c\Vdash A$.
Since $\Vdash\square p\equiv A$, $\+M,a\Vdash\square p$ if and only if $\+M,c\Vdash\square p$.
This contradicts the facts that $\+M,a\not\Vdash\square p$ and $\+M,c\Vdash\square p$.

$(2)$~For the sake of the contradiction, suppose there exists a $\Diamond$-free formula $A$ such that $\Diamond p\equiv A$ is in \Derivable.
Without loss of generality, we may assume that $p$ is the only atomic proposition that may occur in $A$.
Since $\Diamond p\equiv A$ is in \Derivable, by Theorem~\ref{theorem:soundness:dfik}, $\Vdash\Diamond p\equiv A$.
Let $(W,\leq,R,V)$ be the bi-relational model defined by $W=\{a,b,c,d\}$, $a\leq c$, $b\leq d$, $aRb$, $cRb$, $cRd$ and $V(p)=\{d\}$.
By induction on the $\Diamond$-free formula $B$, the reader may easily verify that $\+M,a\Vdash B$ if and only if $\+M,c\Vdash B$.
Since $A$ is $\Diamond$-free, $\+M,a\Vdash A$ if and only if $\+M,c\Vdash A$.
Since $\Vdash\Diamond p\equiv A$, $\+M,a\Vdash\Diamond p$ if and only if $\+M,c\Vdash\Diamond p$.
This contradicts the facts that $\+M,a\not\Vdash\Diamond p$ and $\+M,c\Vdash\Diamond p$.
\end{proof}

{\noindent \sf \textbf{Lemma~\ref{lemma-sec3-dp}.}} 
Suppose that a sequent $S =~\Rightarrow A_1, \ldots, A_m, \langle G_1\rangle, \ldots, \langle G_n\rangle$ is provable in \calculus, where $A_i$'s are formulas. Then either for some $A_i$, $\Rightarrow A_i$ is provable in \calculus~or for some $G_j$, $\Rightarrow \langle G_j\rangle$ is provable in \calculus.  

\begin{proof}
  By induction on the height of a proof of $S$. If $S$ is an axiom, then some $\Rightarrow\langle G_j\rangle$ must be an axiom. 
  Otherwise $S$ it is obtained by applying a rule to some $A_i$ or to some $\langle G_j\rangle$. In the first case, suppose that $S$ is derived by applying a rule to $A_1$ (to simplifying indexing). We only illustrate two cases: let $A_1 = B\land C$, then we have
  $$\frac{\Rightarrow B, A_2, \ldots, A_m, \langle G_1\rangle, \ldots, \langle G_n\rangle \quad \Rightarrow C, A_2, \ldots, A_m, \langle G_1\rangle, \ldots, \langle G_n\rangle}{\Rightarrow B\land C, A_2, \ldots, A_m, \langle G_1\rangle, \ldots, \langle G_n\rangle}$$
  By induction hypothesis on the first premise either form some $A_i$ ($i=2,\ldots, m$) $\Rightarrow A_i$ is derivable or  some $\langle G_j\rangle$ is derivable and we are done: otherwise  $\Rightarrow B$ must be derivable; in this case by induction hypothesis on the second premise  $\Rightarrow C$ must be derivable; then we conclude by an application of $(\land_R)$. Suppose that $A_1 = \Box B$ and is derived by 
  $$\frac{\Rightarrow A_2, \ldots, A_m, \langle \Rightarrow [\Rightarrow B]\rangle,\langle G_1\rangle, \ldots, \langle G_n\rangle}{\Rightarrow \Box B, A_2, \ldots, A_m, \langle G_1\rangle, \ldots, \langle G_n\rangle}$$
  By induction hypothesis, as before either form some $A_i$ ($i=2,\ldots, m$), $\Rightarrow A_i$ is derivable or some $\Rightarrow\langle G_j\rangle$ is derivable and we are done; otherwise $\Rightarrow \langle \Rightarrow [\Rightarrow B]\rangle$ and by an application of $(\Box_R)$ we conclude.
  If $S$ is derived by applying a rule to some $\Rightarrow\langle G_j\rangle$ the reasoning is the same.
\end{proof}

{\noindent \sf \textbf{Proposition~\ref{prop:dp}.}} 
For any formulas $A,B$, if $\Rightarrow A \lor B$ is provable in \calculus, then either $\Rightarrow A$ or $\Rightarrow  B$ is provable.

\begin{proof}
  Let $\Rightarrow A \lor B$ be provable in \calculus. Then it must be derived by $(\lor_R)$ from $\Rightarrow A, B$ and then we apply the previous lemma.  
\end{proof}

{\noindent \sf \textbf{Lemma~\ref{lemma:star-operator}.}} 
Let $\mathcal{M}=(W,\leq,R,V)$ be a bi-relational model and $x, x'\in W$ with $x\leq x'$. Let $S = \Gamma \Rightarrow \Delta$ be any sequent, if $x \not\Vdash \Delta$ then $x' \not\Vdash \Delta^*$.

\begin{proof}
By induction on the structure of $\Delta^*$. If $\Delta^* = \emptyset$ it follows by definition. Otherwise $\Delta^*= [\Phi_1\Rightarrow\Psi_1^*], \ldots, [\Phi_k\Rightarrow\Psi_k^*]$ where $\Delta=\Delta_0,[\Phi_1\Rightarrow\Psi_1], \ldots, [\Phi_k\Rightarrow\Psi_k]$ and $\Delta_0$ is $[\cdot]$-free.
By hypothesis $x\not\Vdash \Delta$, thus $x\not\Vdash [\Phi_i\Rightarrow\Psi_i]$ for $i = 1,\ldots, k$. Therefore there are $y_1, \ldots, y_k$ with $Rx y_i$ for $i = 1,\ldots, k$ such that $y_i\not\Vdash \Phi_i\Rightarrow\Psi_i$. This means that (a) $y_i \Vdash C$ for every $C\in \Phi_i$ and (b) $y_i\not\Vdash \Psi_i$. By (FC) property there are $y'_1, \ldots, y'_k$ such that $Rx' y'_i$ and $y'_i \geq y_i$ for $i = 1,\ldots, k$. By (a)  it follows that (c) $y'_i \Vdash C$ for every $C\in \Phi_i$; moreover by induction hypothesis it follows that (d) $y'_i\not\Vdash \Psi_i^*$. Thus from (c) and (d) we have $y'_i \not \Vdash \Phi_i\Rightarrow\Psi_i^*$, whence $x'\not\Vdash [\Phi_i\Rightarrow\Psi_i^*]$ for for $i = 1,\ldots, k$, which means that $x'\not\Vdash \Delta^*$.
\end{proof}

{\noindent \sf \textbf{Lemma~\ref{forcing-preserving}.}} 
Given a model $\mathcal{M} = (W,  \leq, R, V)$ and $x\in W$, for any rule ($r$) of the form $\frac{G\{S_1\} \quad G\{S_2\}}{G\{S\}}$ or $\frac{G\{S_1\}}{G\{S\}}$, if $x\Vdash G\{S_i\}$, then $x\Vdash G\{S\}$. 
  
\begin{proof}
  We  proceed by induction on the structure of the context $G\{ \ \}$.
  \begin{itemize}
  \item (base of the induction)  $G\{  \ \} = \emptyset$. We check rule by rule. 
  As an example, we consider $(\Box_R)$ and (inter) rules, the other cases are similar or simpler and are left to the reader. For $(\Box_R)$, suppose by absurdity that 
  $x\Vdash \Gamma \Rightarrow \Delta, \langle \seq [\seq B]\rangle$ but $x\not\Vdash \Gamma \Rightarrow \Delta, \Box B$. It follows that: 
  $x\Vdash A$ for every $A\in \Gamma$, $x\not\Vdash \Delta$, (i) $x\not\Vdash \Box B$, (ii) $x\Vdash \langle \seq [\seq B]\rangle$. From (i) it follows that there is $x_1 \geq x$ and $y_1$, with $R x_1 y_1$ such that $y_1\not\Vdash B$. From (ii) it follows that for all $x'\geq x$ and for all $y$ with $Rx' y$, it holds $y\Vdash B$, thus taking $x' = x_1$ and $y = y_1$ we have a contradiction. 
  
  For (inter)  suppose by absurdity that 
  $x\Vdash \Gamma \Rightarrow \Delta, \langle \Sigma \seq \Pi, [\Lambda \seq \Theta^*]\rangle, [\Lambda \seq \Theta] $ but $x\not\Vdash \Gamma \Rightarrow \Delta, \langle \Sigma \seq \Pi\rangle, [\Lambda \seq \Theta]$. It follows that (i) $x\not\Vdash \langle \Sigma \seq \Pi\rangle$, (ii) $x\not\Vdash [\Lambda \seq \Theta]$, but (iii) $x\Vdash \langle \Sigma \seq \Pi, [\Lambda \seq \Theta^*]\rangle$. By (i) there is $x_1\geq x$, such that $x_1\not\Vdash  \Sigma \seq \Pi$, by (ii) there is $y$ with $Rx y$ such that (iv) $y\not\Vdash \Lambda \seq \Theta$. By (FC) condition, there is $y_1$ such that $R x_1 y_1$ and $y_1 \geq y$.
  By (iii), it follows $x_1 \Vdash \Sigma \seq \Pi, [\Lambda \seq \Theta^*]$ whence (v) $y_1 \Vdash \Lambda \seq \Theta^*$. By (iv) we have that $y \Vdash B$ for every $B\in \Lambda$ and $y\not \Vdash \Theta$. Since $y_1 \geq y$, we have that also $y_1 \Vdash B$ for every $B\in \Lambda$, so that by (v) it must be $y_1\Vdash \Theta^*$. Thus we have $y_1 \geq y$, $y\not \Vdash \Theta$, and $y_1\Vdash \Theta^*$, by the previous lemma we have a contradiction. 
  
  \item (inductive step) Let $G\{ \ \} = \Gamma\Rightarrow\Delta, \langle G'\{\}\rangle$. Let us consider for instance a rule $\frac{G\{S_1\}   \quad G\{S_2\}}{G\{S\}}$. Suppose that $x\Vdash G\{S_1\}$ and $x\Vdash G\{S_2\}$. This means that $x\Vdash \Gamma\Rightarrow\Delta, \langle G'\{S_1\}\rangle$ and $x\Vdash \Gamma\Rightarrow\Delta, \langle G'\{S_2\}\rangle$. We prove that $x\Vdash \Gamma\Rightarrow\Delta, \langle G'\{S\}\rangle$. If $x\not\Vdash B$ for some $B\in \Gamma$, or $x\Vdash {\cal O}$ for some ${\cal O}\in \Delta$ we are done. Otherwise, it must be $x\Vdash \langle G'\{S_1\}\rangle$ and 
  $x\Vdash \langle G'\{S_2\}\rangle$. From this it follows that for all $x'\geq x$, we have $x'\Vdash G'\{S_1\}$ and $x'\Vdash G'\{S_2\}$, by induction hypothesis we get $x'\Vdash G'\{S\}$ and the conclusion follows. 
  
  The case $G\{ \ \} = \Gamma\Rightarrow\Delta, [G'\{\}]$ is similar.
  \end{itemize}
\end{proof}

{\noindent \sf \textbf{Proposition~\ref{sequent-inclusion}.}} 
Let $\Gamma\Rightarrow\Delta$ be a sequent saturated with respect to both (trans) and (inter). If $\Delta$ is of form $\Delta',\langle\Sigma\Rightarrow\Pi\rangle$, then $\Gamma\Rightarrow\Delta\subseteq^\-S\Sigma\Rightarrow\Pi$.

\begin{proof}
  We show this by induction on the structure of $\Delta'$. 
\begin{description}
    \item[Base case] Assume $\Delta'$ is $[\cdot]$-free, then according to Definition \ref{def:structural-inclusion}, it suffices to check $\Gamma\subseteq\Sigma$. Since $\Delta',\langle\Sigma\Rightarrow\Pi\rangle$ is saturated, by the saturation condition associated with (trans), we see that $\Gamma\subseteq\Sigma$.
    \item[Inductive step] Assume $\Delta'$ contains $[\cdot]$ blocks, take an arbitrary $[\Phi\Rightarrow\Psi]$ from it. Then $\Delta$ can be written explicitly as $\Delta'',\langle\Sigma\Rightarrow\Pi\rangle,[\Phi\Rightarrow\Psi]$. By the saturation condition associated with (inter), there is a modal block occurring in $\Pi$ of form $\Omega\Rightarrow\Xi$ s.t. $\Phi\Rightarrow\Psi\subseteq^\-S\Omega\Rightarrow\Xi$. $\Sigma\Rightarrow\Pi$ can be written explicitly as $\Sigma\Rightarrow\Pi',[\Omega\Rightarrow\Xi]$, and further $\Gamma\Rightarrow\Delta$ is $\Gamma\Rightarrow\Delta'',\langle\Sigma\Rightarrow\Pi',[\Omega\Rightarrow\Xi]\rangle,[\Phi\Rightarrow\Psi]$. 
    
    Recall the whole sequent $\Gamma\Rightarrow\Delta$ is saturated with both (trans) and (inter), so is $\Gamma\Rightarrow\Delta'',\langle\Sigma\Rightarrow\Pi',[\Omega\Rightarrow\Xi]\rangle$. By IH, we see that $\Gamma\Rightarrow\Delta'',\langle\Sigma\Rightarrow\Pi',[\Omega\Rightarrow\Xi]\rangle\subseteq^\-S\Sigma\Rightarrow\Pi',[\Omega\Rightarrow\Xi]$. Since $[\Phi\Rightarrow\Psi]$ is arbitrary, by Definition \ref{def:structural-inclusion}, we see that $\Gamma\Rightarrow\Delta'',\langle\Sigma\Rightarrow\Pi',[\Omega\Rightarrow\Xi]\rangle,[\Phi\Rightarrow\Psi]\subseteq^\-S \Sigma\Rightarrow\Pi',[\Omega\Rightarrow\Xi]$ as well.
  \end{description}
  As a result, we conclude $\Gamma\Rightarrow\Delta\subseteq^\-S\Sigma\Rightarrow\Pi$.
\end{proof}

{\noindent \sf \textbf{Proposition~\ref{bound}.}} 
Let $\+Do$ be a derivation with root a non-axiomatic sequent $T = \Gamma \Rightarrow \Delta$ obtained by applying R1-rules to  $\Gamma \Rightarrow \Delta^\sharp$, then any $T'$ occurring in $\+Do$  has size $O(|T|^{|T|+1})$. 

\begin{proof}
Let  $T'$ be any sequent occurring in $\+Do$ We first prove that $md(T') = md(T)$. This is proved by induction on the depth of $T'$ in $\+Do$: the base is $T' = T$, whence trivial; for the inductive step let the claim holds for the conclusion of a rule $(R)$, we prove that  it holds for its premise(s), one of which is $T'$. As an example we show the case of $(\Box_L)$. Let $T_1 = \Sigma, \Box A \Rightarrow \Pi, [\Phi \Rightarrow \Psi]$, so that $T' = \Sigma, \Box A \Rightarrow \Pi, [\Phi, A \Rightarrow \Psi]$. 
We have that $md(T_1) = \max(\max(md(\Sigma), md(A) +1), md(\Pi), \max(md(\Phi), md(\Psi))+1) = \max(md(\Sigma), md(A) +1, md(\Pi), md(\Phi) +1, md(\Psi)+1) =  \max(\max(md(\Sigma), md(A) +1), md(\Pi), \max(md(\Phi), md(\Psi), A)+1)= md(T')$. The other cases are similar. 

Let  $\+T_{T'}$ be the tree associated to  $T'$. By the previous claim, we have $h(\+T_{T'}) = md(T') = md(T) \leq |T|$. Moreover, each node $N$ of  $\+T_{T'}$,  is  a pair of sets formulas belonging to $Sub(T') \subseteq Sub(T)$ whence $|N| \leq 2\times |T|$. Finally each node $N$ has as children either $[\cdot]$-blocks inherited from the root $T$, whose number is $\leq |T|$, or created by subformulas $\Diamond B\in Sub(T') \subseteq Sub(T)$, whence their number is again $\leq |T|$. In conclusion we have that $\+T_{T'}$ is a tree of height $md(T) = O(|T|)$ where each node has size  $O(|T|)$ and has $O(|T|)$ children, whence $Card(\+T_{T'}) = O(|T|^{|T|})$ so that $|T'| = O(|T|) \times  O(|T|^{|T|}) = O(|T|^{|T|+1})$
\end{proof}

{\noindent \sf \textbf{Lemma~\ref{lemma:invariant}.}} 
Given a sequent $S$ occurring as a leaf of a derivation   $\+D$ with root $\Rightarrow A$: 
\begin{enumerate}
    \item Let $T\in^+ S$, where $T = \Gamma \Rightarrow \Delta$, for every rule (R) if $T$ satisfies the R-saturation condition on  some formulas $A_i$ and/or blocks $\langle T_j\rangle, [T_k]$ {\em before} the execution of (the body of) the repeat loop (lines 3-14), then $T$ satisfies the R-condition on the involved $A_i,\langle T_j\rangle, [T_k]$ {\em after} the execution of it.
    \item Let $T\in^+ S$, if $T$ is blocked in $S$ {\em before} the execution of (the body of) the repeat loop, then it is still so {\em after} it.
\end{enumerate}

\begin{proof}
Concerning 1. it is obvious for all rules (R) except for (trans) and (inter) as the calculus is cumulative. Concerning (trans): suppose $T = \Gamma \Rightarrow \Delta'\langle \Sigma \Rightarrow \Pi\rangle$ and $\Gamma \subseteq \Sigma$  {\em before} the execution of repeat loop, we can suppose that $T$ satisfies this condition because of a previous execution of the repeat loop of Procedure (as the root of $\+D$ does not satisfies it): namely by \textbf{EXP2} executed in line 12. Thus $T$ is already R1-saturated, and this implies that $\Gamma$ cannot be expanded anymore, no matter which rules are applied to $\Sigma \Rightarrow \Pi$, whence the inclusion $\Gamma \subseteq \Sigma$ will always hold, in particular after the execution of the  repeat loop. The reasoning for (inter)-rule is similar: the inclusion  $\Lambda\Rightarrow\Theta\subseteq^\-S\Phi\Rightarrow\Psi$ involved in the saturation condition will be preserved for the same reason (in particular because $\Lambda\Rightarrow\Theta$ is R1-saturated).

Concerning 2.  the procedure checks whether $T$ is blocked in $S$ at line 10, this means that (i) $S$ is  already global R2 saturated (whence also $T$), (ii)  $T$ is blocked in $S$ by some $S_1\in^+S$ in $S$,    (iii) because of  $\in^{\langle\cdot\rangle}$-minimality,  for all $S'\in^+ S$, such that $T\in^{\langle\cdot\rangle} S'$, we have that $S'$ is R3-saturated, thus no rule can further modify neither $S'$, nor $T$ (nor $S_1$) during the execution  of (the body of) the procedure. Thus $T$ will be still  blocked in $S$ after it. 
\end{proof}

{\noindent \sf \textbf{Lemma~\ref{lem:eq-class-finite}.}} 
Given a formula $A$,  let $\textbf{Seq}(A)$ be the set of sequents that may occur in any possible derivation with root  $\Rightarrow A$. Let $\textbf{Seq}(A)/_{\simeq}$ be quotient of $\textbf{Seq}(A)$ with respect to block-equivalence $\simeq$ as defined in Definition \ref{def:block-eq}. Then $\textbf{Seq}(A)/_{\simeq}$ is finite.

\begin{proof}
  First observe that block-equivalence $\simeq$ is defined by means of  the $\sharp$-images of two sequents, thus it suffices to show that the set $\Phi_A:=\{\Gamma\Rightarrow\Delta^\sharp~|~\Gamma\Rightarrow\Delta\in\textbf{Seq}(A)\}$ is finite. By proposition \ref{bound} we know that every sequent $\Gamma\Rightarrow\Delta^\sharp\in \Phi_A$ has a bounded size, (namely $O(|A|^{|A|+1})$). Moreover observe that $Sub(\Gamma\Rightarrow\Delta^\sharp) \subseteq Sub(A)$. Thus there may be only finitely-many distinct $\Gamma\Rightarrow\Delta^\sharp$, that is  $\Phi_A$ is finite.
\end{proof}

{\noindent \sf \textbf{Theorem~\ref{termination}.}} 
Let $A$ be a formula. Proof-search for the sequent $\Rightarrow A$ terminates with a finite derivation in which any leaf is either an axiom or global-saturated.

\begin{proof}
  (Sketch) We prove that PROCEDURE($A$) terminates producing a finite derivation, in this case all leaves are  axioms or global-saturated. A non-axiomatic leaf $S$ is necessarily global-saturated, otherwise $S$ would be further expanded in Step 8 of PROCEDURE($A$) and it would not be a leaf.  Thus it suffices to  prove that the procedure produces a finite  derivation. Let  $\+D$ built by PROCEDURE($A$).  First we claim that all branches of  $\+D$ are finite.   Suppose for the sake of a contradiction that $\+D$ contains an infinite branch ${\cal B} = S_0,\ldots, S_i, \ldots $ , with $S_0 = \Rightarrow A$. The branch is generated by  applying repeatedly $\textbf{EXP1}(\cdot), \textbf{EXP2}(\cdot)$ and $\textbf{EXP3}(\cdot)$ to each   $S_i$ (or more precisely to some $T_i\in^+ S_i$) . Since each one of these sub-procedures terminates,  the three of them must infinitely alternate on the branch. By (invariant) Lemma, if $T_i\in^+ S_i$ satisfies a saturation condition for a rule (R) or is blocked in ($S_i$) it will remain so in all $S_j$ with $j >i$. That is to say, further steps in the branch cannot "undo" a fulfilled saturation condition or "unblock" a blocked sequent. We can conclude that the branch must contain infinitely many phases of $\textbf{EXP3}(\cdot)$ each time applied to an unblocked sequent in some $S_i$. This entails that  ${\cal B}$ contains infinitely many sequents that are not $\simeq$-equivalent, but this contradicts previous lemma \ref{lem:eq-class-finite}.
  Thus each branch of the derivation  $\+D$ built by PROCEDURE($A$) is finite. To conclude the proof, just observe that  $\+D$ is  a tree whose branches have a finite length and is finitely branching (namely each node/sequent has at most 2 successors, as the rules of \textbf{C}\calculus~are at most binary), therefore $\+D$ is finite. 
\end{proof}

{\noindent \sf \textbf{Proposition~\ref{hp-fc-property}.}} 
The countermodel $\+M_S$ determined by a global-saturated $S$ is a bi-relational model satisfying the hereditary property(HP) and forward confluence(FC).

\begin{proof}
  In the following proof, we abbreviate $R_S,\leq_S$ as $R$ and $\leq$ respectively for readability. 

  For (HP), take arbitrary $x_{S_1},x_{S_2}\in W_S$ with $x_{S_1}\leq x_{S_2}$. Suppose $S_1,S_2$ are of form $\Gamma_1\Rightarrow\Delta_1$ and $\Gamma_2\Rightarrow\Delta_2$ respectively, then $\Gamma_1\Rightarrow\Delta_1\subseteq^\-S\Gamma_2\Rightarrow\Delta_2$. 
  By definition, it follows $\Gamma_1\subseteq\Gamma_2$. As $V_S(x_{S_1})=\{p~|~p\in\Gamma_1\}$ and $V_S(x_{S_2})=\{p~|~p\in\Gamma_2\}$, we have $V_S(x_{S_1})\subseteq V_S(x_{S_2})$.

  For (FC), take arbitrary $x_{\Gamma\Rightarrow\Delta},x_{\Sigma\Rightarrow\Pi},x_{\Lambda\Rightarrow\Theta}\in W_S$ with $x_{\Gamma\Rightarrow\Delta}\leq x_{\Sigma\Rightarrow\Pi}$ and $Rx_{\Gamma\Rightarrow\Delta}x_{\Lambda\Rightarrow\Theta}$, our goal is to find some $x_0\in W_S$ s.t. both $x_{\Lambda\Rightarrow\Theta}\leq x_0$ and $Rx_{\Sigma\Rightarrow\Pi}x_0$ hold. 
  Since $Rx_{\Gamma\Rightarrow\Delta}x_{\Lambda\Rightarrow\Theta}$, by the definition of $R$, we see that $[\Lambda\Rightarrow\Theta]\in\Delta$ and hence $\Gamma\Rightarrow\Delta$ can be written explicitly as $\Gamma\Rightarrow\Delta',[\Lambda\Rightarrow\Theta]$. 
  Meanwhile, since $x_{\Gamma\Rightarrow\Delta}\leq x_{\Sigma\Rightarrow\Pi}$, by the definition of $\leq$, we have $\Gamma\Rightarrow\Delta',[\Lambda\Rightarrow\Theta]\subseteq^\-S\Sigma\Rightarrow\Pi$. 
  By the definition of structural inclusion, there is a block $[\Phi\Rightarrow\Psi]\in\Pi$ s.t. $\Lambda\Rightarrow\Theta\subseteq^\-S\Phi\Rightarrow\Psi$, and then $\Sigma\Rightarrow\Pi$ can be written explicitly as $\Sigma\Rightarrow\Pi',[\Phi\Rightarrow\Psi]$. Since $\Phi\Rightarrow\Psi\in^+\Sigma\Rightarrow\Pi\in^+S$ and $\in^+$ is transitive, we see that $x_{\Phi\Rightarrow\Psi}\in W_S$ as well. 
  Take $x_{\Phi\Rightarrow\Psi}$ to be $x_0$, by the construction of $\+M_S$, it follows directly $x_{\Lambda\Rightarrow\Theta}\leq x_0$ and $Rx_{\Sigma\Rightarrow\Pi}x_0$. 
\end{proof}

{\noindent \sf \textbf{Lemma~\ref{truth-lemma-formula}.}} 
Let $S$ be a global-saturated sequent and $\+M_{S}$ be defined as above. (a). If $A\in\Phi$, then $\+M_{S},x_{\Phi\Rightarrow\Psi}\Vdash A$; (b). If $A\in\Psi$, then $M_{S},x_{\Phi\Rightarrow\Psi}\nVdash A$.

\begin{proof}
  We prove the lemma by induction on the complexity of $A$. For convenience, we abbreviate $x_{\Phi\Rightarrow\Psi},\leq_S,R_S,W_S$ as $x,\leq, R,W$ respectively in the following proof.
  \begin{itemize}
    \item $A$ is of form $p,\bot,\top,B\vee C,B\wedge C$. These cases are similar and relatively trivial, here we only give the proof for $B\wedge C$ as an example. Recall that both R3-saturated and blocked sequents are already R1-saturated, so it is not necessary to distinguish the cases whether $\Phi\Rightarrow\Psi$ is blocked or not.

    For (a), let $B\wedge C\in\Phi$. By saturation  we have that both $B,C\in\Phi$. Thus by IH, we have  $x\Vdash B$ and $x\Vdash C$, whence   $x\Vdash B \land C$.

    For (b), let $B\wedge C\in\Psi$. By saturation  either $B\in\Psi$ or $C\in\Psi$. Thus by IH either  $x\not\Vdash B$ and $x\not\Vdash C$ hold. In both cases we get $x\not\Vdash B\wedge C $.

    \item $A$ is of form $B\supset C$. For (a), let $B\supset C\in\Phi$. 
    Assume for the sake of a contradiction that $x\nVdash B\supset C$. Then there exists a world $x_{0} =  x_{\Sigma\rightarrow \Pi}\in W_S$, with  $x\leq x_0$ such that $x_{0}\Vdash B$ and $x_{0}\nVdash C$.  By IH, we have $B\notin\Pi$ and $C\notin\Sigma$. Meanwhile, since  $\Sigma\Rightarrow\Pi$ satisfies the saturation condition associated with $(\supset_L)$ (no matter whether is blocked or not),  either $B\in\Pi$ or $C\in\Sigma$, and we have a contradiction.

    For (b), let $B\supset C\in\Psi$. We distinguish whether $\Phi\Rightarrow\Psi$ is blocked sequent or not. 
    Assume first that $\Phi\Rightarrow\Psi$ is not blocked, then it satisfies one of the two saturation conditions associated with $(\supset_R)$ for $B\supset C$: 
    \begin{enumerate}
      \item[(1).] $B\in \Phi$ and $C\in\Psi$. In this case by IH, it follows $x\Vdash B$ and $x\nVdash C$. By reflexivity  $x\leq x$, we conclude $x\not\Vdash B\supset C$.
      \item[(2).] there is a block $\langle\Lambda\Rightarrow\Theta\rangle\in\Psi$ s.t. $B\in\Lambda$ and $C\in\Theta$.  By saturation (and Proposition \ref{sequent-inclusion}), we have $\Phi\Rightarrow\Psi\subseteq^\-S\Lambda\Rightarrow\Theta$, whence  $x\leq x_{\Lambda\Rightarrow\Theta}$. Since $B\in\Lambda$ and $C\in\Theta$, by IH, we have $x_{\Lambda\Rightarrow\Theta}\Vdash B$ and  $x_{\Lambda\Rightarrow\Theta}\nVdash C$, thus $x\not\Vdash B\supset C$. 
    \end{enumerate}
    Assume now that $\Phi\Rightarrow\Psi$ is blocked and it does not satisfy the previous condition (1),  otherwise we conclude the proof as before. 
    By definition, there is an unblocked sequent $\Sigma\Rightarrow\Pi\in^+S$ s.t. $\Phi\Rightarrow\Psi$ is blocked by it. 
    Then we have $\Sigma\Rightarrow\Pi\simeq \Phi\Rightarrow\Psi$, which implies $\Pi^\sharp=\Psi^\sharp$, thus also $B\supset C\in\Pi$. 
    Observe that $\Sigma\Rightarrow\Pi\simeq \Phi\Rightarrow\Psi$ implies  $\Phi\Rightarrow\Psi\subseteq^\-S\Sigma\Rightarrow\Pi$ hold,   thus (*)  $x\leq x_{\Sigma\Rightarrow\Pi}$ by model construction.
    Given that $\Sigma\Rightarrow\Pi$ is R3-saturated, it already satisfies the saturation condition associated with $(\supset_R)$ for $B\supset C$. Since $\Sigma\Rightarrow\Pi\simeq \Phi\Rightarrow\Psi$, we get that $\Sigma\Rightarrow\Pi$ does not satisfy condition (1), thus it satisfies condition (2), that is there is   there is a block $\langle\Lambda\Rightarrow\Theta\rangle\in\Pi$ such that $B\in\Lambda$ and $C\in\Theta$. We have $\Sigma\Rightarrow\Pi\subseteq^\-S\Lambda\Rightarrow\Theta$, whence $x_{\Sigma\Rightarrow \Pi}\leq x_{\Lambda\Rightarrow\Theta}$ so that  by (*) and transitivity also $x\leq x_{\Lambda\Rightarrow\Theta}$. Then we proceed as in case (2) above. 

    \item $A$ is of form $\square B$. 
    For (a), let $\square B\in\Phi$. Similar as the ($\supset$)-case, $\Phi\Rightarrow\Psi$ satisfies the saturation condition associated with $(\square_R)$ for $\square B$ regardless of whether the sequent itself is blocked or not. Assume for the sake of a contradiction that $x\nVdash \square B$. 
    Then there exists $x_{\Sigma\Rightarrow\Pi},x_{\Lambda\Rightarrow\Theta}$ denoted as $x_1,x_2$ s.t. $x\leq x_1, Rx_1x_2$ and $x_2\nVdash B$. By IH, we see that $B\notin\Lambda$. Meanwhile, according to the model construction, we see that $\Phi\Rightarrow\Psi\subseteq^\-S\Sigma\Rightarrow\Pi$ and $[\Lambda\Rightarrow\Theta]\in\Pi$. Moreover we have $\Phi\subseteq\Sigma$, thus $\square B\in\Sigma$ as well. Also, since $\Sigma\Rightarrow\Pi$ is of form $\Sigma\Rightarrow\Pi',[\Lambda\Rightarrow\Theta]$, by the saturation condition associated with $(\square_L)$, we have $B\in \Lambda$, which leads to a contradiction.

    For (b), let $\square B\in \Psi$. We distinguish whether $\Phi\Rightarrow\Psi$ is blocked or not. Assume that  $\Phi\Rightarrow\Psi$ is not blocked, then it satisfies the one of the two saturation conditions associated with $(\square_R)$ for $\square B$: 
    \begin{enumerate}
      \item[(1).] there is a block $[\Lambda\Rightarrow\Theta]\in\Psi$ with $B\in\Theta$. By IH, we have $x_{\Lambda\Rightarrow\Theta} \nVdash B$. 
      By reflexivity  $x\leq x$ and model construction  $Rxx_{\Lambda\Rightarrow\Theta}$, so that  $x \nVdash \Box B$. 
      \item[(2).] there is a block $\langle\Omega\Rightarrow[\Lambda\Rightarrow\Theta],\Xi\rangle\in\Psi$ with $B\in\Theta$.  Denote the sequent $\Omega\Rightarrow[\Lambda\Rightarrow\Theta],\Xi$ by $S_0$. Since $\Phi\Rightarrow\Psi$ is saturated with (trans) and (inter), by Proposition \ref{sequent-inclusion}, we have $\Phi\Rightarrow\Psi\subseteq^\-S S_0$. According to the model construction, we see that $x\leq x_{S_0}$ and  $Rx_{S_0}x_{\Lambda\Rightarrow\Theta}$. Since $B\in\Theta$, by IH we have $x_{\Lambda\Rightarrow\Theta} \nVdash B$ and we can conclude   $x \nVdash \Box B$. 
    \end{enumerate}

    Assume that  $\Phi\Rightarrow\Psi$ is blocked and does not satisfy condition (1) for $\Box B$, otherwise the proof proceeds as in case (1) above. Then there is an unblocked sequent $\Sigma\Rightarrow\Pi\in^+S$ such that $\Phi\Rightarrow\Psi$ is blocked by it. Then $\Sigma\Rightarrow\Pi\simeq \Phi\Rightarrow\Psi$, which implies $\Pi^\sharp=\Psi^\sharp$, so $\square B\in\Pi$ as well. 
    Moreover, by definition, we have $\Phi\Rightarrow\Psi\subseteq^\-S\Sigma\Rightarrow\Pi$, whence  by  model construction (**) $x\leq x_{\Sigma\Rightarrow\Pi}$. 
    Given that $\Sigma\Rightarrow\Pi$ is R3-saturated, it satisfies the saturation condition associated with $(\square_R)$ for $\square B$,  but since $\Sigma\Rightarrow\Pi\simeq \Phi\Rightarrow\Psi$, we have that $\Sigma\Rightarrow\Pi$ does not satisfy condition (1), thus it must satisfy condition (2).  Therefore there is 
    there is a block $\langle\Omega\Rightarrow[\Lambda\Rightarrow\Theta],\Xi\rangle\in\Pi$,  such that $B\in \Theta$. Letting $S_0 = \Omega\Rightarrow[\Lambda\Rightarrow\Theta],\Xi$, we have  $x_{\Sigma\Rightarrow\Pi}\leq x_{S_0}$ and $Rx_{S_0}x_{\Lambda\Rightarrow\Theta}$. By (**) we have also  $x\leq x_{S_0}$ and we conclude as in case (2) above. 
    \item $A$ is of form $\Diamond B$. It is not necessary to distinguish cases when $\Phi\Rightarrow\Psi$ is blocked or not.

    For (a), let $\Diamond B\in\Phi$. Then by the saturation condition associated with $(\Diamond_L)$, there is a block $[\Lambda\Rightarrow\Theta]\in\Psi$ s.t. $B\in\Lambda$.  By model construction, we have $Rxx_{\Lambda\Rightarrow\Theta}$ and by IH, we get  $x_{\Lambda\Rightarrow\Theta}\Vdash B$, thus  $x\vDash \Diamond B$. 

    For (b), let $\Diamond B\in \Psi$. Let $y\in W$, with $Rx y$ we show that $y\nVdash B$.  If $Rx y$ it must be $y = x_{\Lambda\Rightarrow\Theta}$  and $[\Lambda\Rightarrow\Theta]\in\Psi$. By saturation condition for  ($\Diamond_R$), it follows that $B\in \Theta$, thus by IH  $x_{\Lambda\Rightarrow\Theta} \nVdash B$ and we are done. 
  \end{itemize}
This completes our proof.
\end{proof}

{\noindent \sf \textbf{Theorem~\ref{completeness-cfik}}}
For any formula $A\in \+L$, if $\Vdash A$, then $\Rightarrow A$ is provable in \textbf{C}\calculus.

\begin{proof}
By contraposition. Given a formula $A$, if $A$ is unprovable in \textbf{C}\calculus, then we see that PROCEDURE($A$) produces a derivation containing a non-axiomatic global saturated leaf $S = \Gamma \Rightarrow \Delta$ such that $A\in\Delta$. By the truth lemma, $A$ is not valid in the model $\+M_S$. 
\end{proof}

\end{document}